\documentclass[11pt,a4paper]{article}
\usepackage[margin=1in]{geometry}
\usepackage{amssymb,amsthm,mathtools}
\usepackage{enumitem}
\usepackage{booktabs}
\usepackage{float}
\usepackage{microtype}
\usepackage{hyperref}
\usepackage{amssymb,amsmath,mathtools,amsthm}
\usepackage{enumitem,booktabs,float,multirow,array}
\theoremstyle{plain}
\newtheorem{theorem}{Theorem}[section]
\newtheorem{lemma}{Lemma}[section]
\newtheorem{proposition}{Proposition}[section]

\theoremstyle{definition}
\newtheorem{definition}{Definition}[section]
\theoremstyle{remark}

\newtheorem{notation}[theorem]{Notation}

\DeclareMathOperator{\Nm}{N}
\DeclareMathOperator{\Gal}{Gal}
\DeclareMathOperator{\Var}{Var}
\DeclareMathOperator{\Cov}{Cov}

\newcommand{\Z}{\mathbb{Z}}
\newcommand{\Q}{\mathbb{Q}}
\newcommand{\R}{\mathbb{R}}
\newcommand{\C}{\mathbb{C}}

\newcommand{\ri}{\mathrm{i}}

\newcommand{\nm}[1]{\lVert#1\rVert}
\newcommand{\abs}[1]{\lvert#1\rvert}
\newcommand{\ip}[2]{\langle #1,#2\rangle}
\newcommand{\tO}{\tilde{O}}
\newcommand{\CN}{\mathcal{CN}}
\newcommand{\Exp}{\mathrm{Exp}}

\newcommand{\E}{\mathbb{E}}

\title{Module Lattice Security (Part III):\\
{\large Structured CVP Distance on the Log-Unit Lattice}}
\author{Ming-Xing Luo
\\
\footnotesize
School of Information Science and Technology, Southwest Jiaotong University, Chengdu 610031, China}

\begin{document}
\maketitle

\begin{abstract}
We prove that the $L^2$ CVP distance from a random short ring element to the log-unit lattice of $\Q(\zeta_{2^k})$ converges to $\frac{\pi}{2\sqrt{6}}\sqrt{n}$ as $n=2^{k-1}\to\infty$. We then show that this target lies inside the Voronoi cell of the origin for $k\ge 4$. For the $L^\infty$ norm, the maximum over $n$ sub-Gaussian coordinates yields $O(\sqrt{\log n})$ which translates into a sub-polynomial approximation factor for the Short Generator Problem. We show a Coarse Lattice Theorem that Babai's algorithm returns zero for all structured targets, yet exactly recovers unit perturbations of arbitrary size. For module determinant ideals, we further prove the Trigamma Theorem that proves an intrinsic imbalance $\sigma_{g_0}=O(1)$ independent of the modulus $q$. Finally, combined with Parts I and II, we reduce the CDPR factor for ML-KEM from $\exp(\tO(\sqrt{n}))$ to a sub-polynomial value.
\end{abstract}


\section{Introduction}\label{seintro}

The CDPR quantum attack \cite{CDPR16} on ideal lattices over $R=\Z[\zeta_{2^k}]$ achieves an approximation factor $\gamma=\exp(C_0\sqrt{n}\log n)$, where $n=2^{k-1}$, and $C_0>0$ is an absolute constant. Refined by Cramer, Ducas, and Wesolowski \cite{CDW17,CDW21} using Stickelberger relations, this attack implements with three phases. Phase 1 is named as quantum PIP which uses the Biasse-Song algorithm \cite{BiasseSong16} to find some generator $g$ of a given principal ideal $I=(g_0)\subset R$ in quantum polynomial time. Unfortunately, $g$ may be exponentially larger than the shortest generator $g_0$ as $g=g_0\epsilon$ for some unit $\epsilon\in R^\times$ that can have enormous norm. Phase 2 is log-unit CVP, where Closest Vector Problem (CVP) is resolved on the log-unit lattice $\Lambda$ to recover $\epsilon$. Note in the log-Minkowski embedding, the generators of ideal $I$ form the coset $\log\abs{\sigma(g_0)}+\Lambda$. So, the problem of finding a short generator is then reduced to locate the lattice point closest to $\log\abs{\sigma(g)}$. In Phase 3, the CVP solution can be used to identify a unit $\epsilon'$ that closes to $\epsilon$. Meanwhile, the short generator is recovered as $g'=g(\epsilon')^{-1}$. In this case, the approximation factor is given by $\gamma=\exp(\rho_\infty)$, where $\rho_\infty$ denotes the $L^\infty$-norm of the CVP residual.

Pellet-Mary, Hanrot, and Stehl\'{e} \cite{PMHS19} gave a time-approximation tradeoff with preprocessing. Felderhoff \textit{et al.} \cite{FPSW23} proved Ideal-SVP remains hard for small-norm prime ideals. For module lattices, Langlois and Stehl\'{e} \cite{LS15} proved the worst-case MLWE hardness. Mureau \textit{et al.} \cite{MPPW24} and Allombert \textit{et al.} \cite{APvW25} attacked rank-2 module-LIP. Chevignard \textit{et al.} \cite{CMEPW25} reduced Hawk scheme to PIP. Ducas \textit{et al.} \cite{DEP25} predicted module-BKZ quality. Babai \cite{Babai86} introduced the nearest-plane CVP algorithm, while refs.\cite{LLL82,Schnorr87,GN08,ADRS15} provide generic lattice reduction with progressively better approximation factors.

Note the worst-case bound treats the CVP target as an arbitrary point in the ambient space. However, the targets arising from short generators are highly structured: they are log-embeddings of ring elements with small coefficients. A natural question is whether this structure can be exploited to obtain tighter bounds. Our goal in what follow is to answer this problem.

Our probabilistic methods use classical central limit theorem (CLT) (Lindeberg \cite{Billingsley95,Feller71}) and sub-Gaussian concentration inequalities from Vershynin \cite{Vershynin18} and Boucheron-Lugosi-Massart \cite{BLM13}, the QR decomposition of random matrices and the associated $\chi$-squared distributions from Tao \cite{Tao12}, and the trigamma functions from Abramowitz-Stegun \cite{AS64}. Specifically, we prove the following results.
\begin{itemize}
  \item \textbf{Theorem \ref{thm:A}}
    For a short element $g\in R$ with i.i.d. bounded coefficients, its $L^2$ distance from $\Pi_{H_0}(\log\abs{\sigma(g)})$ to $\Lambda$ converges to $\frac{\pi}{2\sqrt{6}}\sqrt{n}\approx 0.6413\sqrt{n}$. 
  \item \textbf{Theorem \ref{thm:C}}
    The structured target vector lies inside the Voronoi cell of the origin in $\Lambda$ with probability $1-o(1)$ for $k\ge 4$. 
  \item \textbf{Theorem \ref{thm:B}}
    The $L^\infty$ distance is $O(\sqrt{\log n})$ for short generators, which implies the approximation factor $\gamma=\exp(O(\sqrt{\log n}))=o(n^\epsilon)$ for every $\epsilon>0$. For ML-KEM scheme ($n=256$), $\gamma\approx 2^{4.1}$ (Gaussian) or $2^{4.9}$ (ring), far below the threshold.
\end{itemize}

The rest of this paper is organized as follows. Section \ref{seprelim} introduces some preliminaries of the log-unit lattice, Babai's nearest-plane algorithm, and complex Gaussians. Section \ref{seclt} contributes the embedding uncorrelatedness, central limit theorem, and variance identity. Section \ref{secvp} presents three CVP distance Theorems. Section \ref{secoarse} shows the Coarse Lattice Theorem. Section \ref{setrigamma} provides the Trigamma Theorem and its implications for ML-KEM security while Section \ref{seconclusion} concludes the paper.

\section{Preliminaries}
\label{seprelim}

This section introduces algebraic, geometric, and probabilistic foundations used in the following sections. 

\subsection{Setup and notation}

The cyclotomic field $K=\Q(\zeta)$ is obtained by adding a primitive $m$-th root of unity $\zeta=e^{2\pi\ri/m}$ to the rational number field. Elements of $K$ are represented by $\sum_{i=0}^{n-1}a_i\zeta^i$ with $a_i\in\Q$. The degree of $K$ is $[K:\Q]=n=\varphi(m)=2^{k-1}$, because $\zeta$ satisfies the minimal polynomial $x^n+1$ over $\Q$ which is valid for $m=2^k$ with $k\ge 3$. The ring of integers $R=\Z[\zeta]$ consists of elements with integer coefficients $a_i\in\Z$.

The Galois group $\Gal(K/\Q)$ consists of all field automorphisms of $K$ that fix the field $\Q$. Since every such automorphism will map $\zeta$ to another primitive $m$-th root of unity, $\Gal(K/\Q)$ can be identified with $(\Z/m\Z)^\times$, i.e., the group of integers modulo $m$ coprime to $m$. For $m=2^k$, these are the odd integers $a\in\{1,3,5, \dots, m-1\}$, and the corresponding automorphism $\sigma_a$ is defined by $\sigma_a(\zeta)=\zeta^a$. This action is well-defined and  can be extended to all of $K$ by linearity.

An embedding $\sigma_j\colon K\hookrightarrow\C$ is a ring homomorphism from $K$ into the complex numbers. Since $K$ is complex for $k\ge 3$, all $n$ embeddings are mapped into $\C\setminus\R$. The conjugate pair of embeddings satisfies $\sigma_{m-j}(g)=\overline{\sigma_j(g)}$ for all $g\in K$, because $\zeta^{m-j}=\overline{\zeta^j}$.

For an element $g=\sum_{i=0}^{n-1}c_i\zeta^i\in R$, the $j$-th embedding can be defined by evaluating the polynomial at $\zeta^j$, i.e., 
\begin{align}\label{alemb}
  \sigma_j(g)=\sum_{i=0}^{n-1}c_i\zeta^{ij}=\sum_{i=0}^{n-1}c_i e^{2\pi\ri\cdot ij/m}.
\end{align}
This is a discrete Fourier transform (DFT) of the vector $(c_0, \dots, c_{n-1})$ evaluated at the $n$-th roots of $-1$ from $\zeta^n=e^{\pi\ri}=-1$.  

\begin{definition}\label{delogemb}
The log-embedding of $g\in R\setminus\{0\}$ is the vector of coordinate-wise log-moduli defined by 
\begin{align}\label{allogemb}
  L(g)=\bigl(\log\abs{\sigma_j(g)}\bigr)_{\text{odd }j}\in\R^n.
\end{align}
\end{definition}

This mapping transforms the multiplicative of $R$ into additive in $\R^n$ as 
\begin{align}\label{allogmult}
  L(g_1 g_2)=L(g_1)+L(g_2),
\end{align}
because $|\sigma_j(g_1 g_2)|=|\sigma_j(g_1)|\cdot |\sigma_j(g_2)|$. Especially, $g=g_0\epsilon$ becomes $L(g)=L(g_0)+L(\epsilon)$ for $g_0\in R\setminus\{0\}$ and $\epsilon\in R^\times$.

The trace-zero hyperplane is defined by $H_0=\{\mathbf{x}\in\R^n: \sum_j x_j=0\}$. The orthogonal projection onto $H_0$ is then defined by 
\begin{align}\label{alproj}
  \Pi_{H_0}(L(g))=L(g)-\frac{1}{n}\log\abs{\Nm_{K/\Q}(g)}\cdot\mathbf{1},
\end{align}
where the field norm $\Nm_{K/\Q}(g)=\prod_{\text{odd }j}\sigma_j(g)$, i,e., the product over all $n$ embeddings indexed by $j\in\{1,3, \dots, m-1\}$. Here, conjugate pairs $\sigma_j$ and $\sigma_{m-j}$ satisfy $\sigma_{m-j}(g)=\overline{\sigma_j(g)}$. 

Denote $J=\{1,3, \dots, n-1\}$ as a set representative for all $\tfrac{n}{2}$ conjugate pairs. Then $\sigma_j(g)\cdot\sigma_{m-j}(g)=|\sigma_j(g)|^2$ for each $j\in J$. So, $N_{K/\Q}(g)$ is real and positive as
\begin{align}\label{alnorm-conj}
  \Nm_{K/\Q}(g)=\prod_{j\in J}|\sigma_j(g)|^2
  =\prod_{\text{odd }j}|\sigma_j(g)|.
\end{align}
The field norm is a rational integer for $g\in R$.

The log-unit lattice $\Lambda$ lies in $H_0$ since units $\epsilon$ satisfy $\abs{\Nm(\epsilon)}=1$ and $\sum_j\log|\sigma_j(\epsilon)|=\log\abs{\Nm(\epsilon)}=0$. So, CVP will be defined in $H_0$. Moreover, using the projection \eqref{alproj}, the CVP distance is independent of the overall scale of $g$.

\subsection{The log-unit lattice}

By Dirichlet's Unit Theorem \cite{Neukirch99}, the unit group of a number field $K$ with $r_1$ real and $r_2$ pairs of complex embeddings has unit rank $r_1+r_2-1$. For $K=\Q(\zeta_{2^k})$ with $k\ge 3$, we have $r_1=0$ and $r_2=\tfrac{n}{2}$. This gives unit rank $r=\tfrac{n}{2}-1$. The log-embedding $L(\epsilon)$ for a unit $\epsilon$ satisfies $\log|\sigma_j(\epsilon)|=\log|\sigma_{m-j}(\epsilon)|$ by conjugation, so we have effective dimension $\tfrac{n}{2}$ and the lattice rank in $H_0$ $r=\tfrac{n}{2}-1$.

Since $h_k^+=1$ for $k\le 12$ (Part I), the cyclotomic units $\xi_a=\frac{\sin(a\pi/m)}{\sin(\pi/m)}$ for odd $a\in\{3, 5, \dots, n-1\}$ generate $R^\times$ modulo torsion \cite{Sinnott78,Wash97}. 

\begin{definition}\label{delogunit}
The log-unit lattice is defined by 
\begin{align}\label{alLambda-def}
  \Lambda=\{(\log|\sigma_j(\epsilon)|)_{\text{odd } j}: \epsilon\in R^\times\} \subset H_0.
\end{align}
Its rank $r=\tfrac{n}{2}-1$ and a basis $\{\mathbf{b}_a\}$ is indexed by $a\in\{3, \dots, n-1\}$, where $\mathbf{b}_a=\Pi_{H_0}((\log|\sigma_j(\xi_a)|)_j)$.
\end{definition}

$\Lambda$ is an $(\tfrac{n}{2}-1)$-dimensional lattice in the $(n-1)$-dimensional hyperplane $H_0\subset\R^n$. Its points correspond to the logarithmic signatures of units. The CVP on $\Lambda$ asks: given the log-signature of an arbitrary generator, find the unit whose log-signature is closest.

\begin{definition}\label{decoverradius}
The $L^p$-cover radius of a lattice $\Lambda\subset\R^m$ is defined by 
\begin{align}\label{alcover-rad}  \mu_p(\Lambda)=\max_{\mathbf{t}\in\R^m/\Lambda} \min_{\mathbf{v}\in\Lambda}\nm{\mathbf{t-v}}_p.
\end{align}
\end{definition}

This radius is the maximum CVP distance over all possible target vectors. The CDPR method \cite{CDPR16} uses $\mu_\infty(\Lambda)=O(\sqrt{n\log n})$. Instead, we show structured targets are much closer than this worst case.

\begin{definition}\label{devoronoi}
The Voronoi cell of a lattice point $\mathbf{v}\in\Lambda$ is defined as 
\begin{align}\label{alvoronoi}
  \mathcal{V}(\mathbf{v})=\{\mathbf{t}\in\R^m:\nm{\mathbf{t-v}}_2\le\nm{\mathbf{t-w}}_2 \mbox{ for all }\mathbf{w}\in\Lambda\}.
\end{align}
\end{definition}

A point $\mathbf{t}$ lies in $\mathcal{V}(0)$ if and only if $2\ip{\mathbf{t}} {\mathbf{w}}\le \|\mathbf{w}\|_2^2$ for every nonzero $\mathbf{w}\in\Lambda$ (the Voronoi condition).

\subsection{Babai's nearest-plane algorithm}

Babai's algorithm \cite{Babai86} is the standard polynomial-time CVP approximation used in Phase 2 of the CDPR attack \cite{CDPR16}.  

\begin{definition}\label{deGS}
Given a basis $B=(\mathbf{b}_1, \dots, \mathbf{b}_r)$ of a lattice $\Lambda\subset\R^m$, the Gram-Schmidt orthogonalization (GSO) produces orthogonal vectors $\mathbf{b}_1', \dots, \mathbf{b}_r'$ as
\begin{align}\label{algso}
  \mathbf{b}_i'=\mathbf{b}_i-\sum_{j<i}\mu_{ij} \mathbf{b}_j',
  \quad
  \mu_{ij}=\frac{\ip{\mathbf{b}_i}{\mathbf{b}_j'}}{\nm{\mathbf{b}_j'}_2^2}.
\end{align}
\end{definition}

For the log-unit lattice, these norms are $\Omega(\sqrt{n})$ (Table \ref{tab:coarse}) which is large relative to $O(1)$ per-component deviation of structured targets.

The basis matrix $B\in\R^{m\times r}$ admits a QR decomposition $B=QR$, where $Q\in\R^{m\times r}$ has orthonormal columns and $R\in\R^{r\times r}$ is an upper triangular with positive diagonal. The two factorizations are related by
\begin{align}\label{alqr-gs}
  Q_i=\frac{\mathbf{b}_i'}{\nm{\mathbf{b}_i'}_2},
  \quad 
  R_{ii}=\nm{\mathbf{b}_i'}_2,
  \quad
  R_{ji}=\mu_{ij}\nm{\mathbf{b}_j'}_2 \quad (j<i).
\end{align}

\begin{definition}[Babai's nearest-plane algorithm \cite{Babai86}]\label{debabai} Given a lattice basis $B=(\mathbf{b}_1, \dots, \mathbf{b}_r)$ with GSO $(\mathbf{b}_1', \dots, \mathbf{b}_r')$ and a target $\mathbf{t}\in\R^m$, Babai's algorithm computes integer coefficients $c_r, c_{r-1}, \dots, c_1$ as 
\begin{align}\label{albabai}
  c_i=\Bigl\lfloor\frac{\ip{\mathbf{t}-\sum_{j>i}c_j \mathbf{b}_j}{\mathbf{b}_i'}}        {\nm{\mathbf{b}_i'}_2^2} \Bigr\rceil,
  \quad i=r, r-1, \dots, 1,
\end{align}
where $\lfloor\cdot\rceil$ denotes rounding to the nearest integer. The output lattice vector is given by $\mathbf{v}=\sum_{i=1}^r c_i \mathbf{b}_i$ with Babai residual $\mathbf{e}=\mathbf{t}-\mathbf{v}$.
\end{definition}

The Babai residual satisfies $\mathbf{e}=\sum_{i=1}^r (\mu_i-c_i) \mathbf{b}_i' +\mathbf{t}^{\perp}$, where $\mathbf{t}^{\perp}$ is the component of $\mathbf{t}$ which is orthogonal to $\mathrm{span}(\mathbf{b}_1, \dots,  \mathbf{b}_r)$, and $\mu_i$ is the projection coefficient computed at step $i$. Moreover, the Babai residual satisfies
\begin{align}\label{alBabaiB}
  \|\mathbf{t}-\mathbf{v}\|_2 \le \frac{\sqrt{r}}{2} \max_{1\le i\le r}
   \|\mathbf{b}_i'\|_2+ \|\mathbf{t}^{\perp}\|_2.
\end{align}
When $\Lambda$ has full column rank of $m=r$, $\mathbf{t}^{\perp}$ vanishes and the bound simplifies into
\begin{align}\label{alBabaiBf}
  \|\mathbf{t}-\mathbf{v}\|_2 \le \frac{\sqrt{r}}{2} \max_{1\le i\le r} \|\mathbf{b}_i'\|_2.
\end{align}

For the log-unit lattice $\Lambda$ with rank $r=\frac{n}{2}-1$ and GS norms $\nm{\mathbf{b}_i'}_2=\Omega(\sqrt{n})$, the Babai bound gives $\|\mathbf{e}\|_2=O(n)$. This is the worst case guarantee over all targets. For structured targets, we will show the projection coefficients $\mu_i$ are not bounded by $\frac{1}{2}$. 

Babai's algorithm costs $O(r^2)$ arithmetic operations on GS coefficients (or $O(r^2 m)$ if the vectors are in $\R^m$). For the log-unit lattice with $r=\frac{n}{2}-1$ and $m=n$, the total cost is $O(n^3)$ which might be negligible compared to the PIP step.

\subsection{Probabilistic definitions}

We collect some probabilistic definitions used in the following proofs. 

\begin{definition}\label{deCG}
A circularly symmetric complex Gaussian $Z\sim\CN(0,\sigma^2)$ is $Z=X+\ri Y$, where $X, Y\sim\mathcal{N}(0,\sigma^2/2)$ are independent. Circular symmetry means $Z$ and $e^{\ri\theta}Z$ have the same distribution for every $\theta$.
\end{definition}

The squared modulus satisfies $|Z|^2/\sigma^2\sim\Exp(1)$ (exponential with rate $1$), and $|Z|$ has Rayleigh distribution with parameter $\sigma/\sqrt{2}$. The circular symmetry of each $\sigma_j(g)=\sum c_i e^{2\pi\ri\cdot ij/m}$ is from the fact that the summands are i.i.d. terms rotated by different angles, uniformly distributed modulo $2\pi$.

\begin{definition}\cite{Vershynin18}\label{desubG}
A centered real random variable $X$ is sub-Gaussian with parameter $\sigma^2$, denoted as $X\sim\mathrm{subG}(0,\sigma^2)$, if the following condition holds 
\begin{align}\label{alsubG}
  \mathbb{E}[e^{tX}]\le e^{\sigma^2 t^2/2}\quad\text{for all }t\in\R.
\end{align}
\end{definition}

Note the sum of independent sub-Gaussian variables is sub-Gaussian, and any bounded variable with $|X|\le B$ is sub-Gaussian with parameter $B^2$ from the Hoeffding's Lemma \cite{Durrett19}.

\begin{definition}\label{dechisq}
If $W_1, \dots, W_k\overset{\text{i.i.d.}}{\sim}\mathcal{N}(0,1)$, then $\sum_{i=1}^k W_i^2$ is the $\chi$-squared distribution ($\chi_k^2$) with $k$ degrees of freedom.
\end{definition}

It satisfies $\mathbb{E}[\chi^2_k]=k$ and $\Var(\chi^2_k)=2k$. The logarithmic moments, which are central to our Trigamma Theorem, involve the digamma and trigamma functions as 
\begin{align}\label{allogchiq}
  \mathbb{E}[\log\chi^2_k]=\psi\bigl(\tfrac{k}{2}\bigr)+\log 2, \quad \Var(\log\chi^2_k)=\psi'\bigl(\tfrac{k}{2}\bigr).
\end{align}

\begin{definition} \cite{AS64}\label{detrigamma}
The digamma function is defined by $\psi(z)=\frac{\mathrm{d}}{\mathrm{d}z}\log\Gamma(z)$ and trigamma function is $\psi'(z)=\frac{\mathrm{d}^2} {\mathrm{d}z^2}\log\Gamma(z)$.
\end{definition}

Note $\psi(1)=-\gamma$, where $\gamma\approx 0.5772$ is the Euler-Mascheroni constant. Moreover, we have $\psi'(z)=\sum_{k=0}^\infty(z+k)^{-2}$, which is positive, strictly decreasing, and tends to $1/z$ for large $z$. For $z=1$, we have $\psi'(1)=\pi^2/6$.

\begin{notation}\label{not:partIII}
We fix the following notation throughout this part:
\begin{itemize}[nosep]
\item $k\ge 3$ is an integer, $n=2^{k-1}$, $m=2^k=2n$, $\zeta=e^{2\pi\ri/m}$.
\item $R=\Z[\zeta]$ is the ring of integers of $K=\Q(\zeta)$.
\item The Galois group $\Gal(K/\Q)\cong(\Z/m\Z)^\times$ acts via $\sigma_a:\zeta\mapsto\zeta^a$ for odd $a$.
\item The $n$ complex embeddings are $\sigma_j$ for odd $j\in\{1,3, \dots, m-1\}$.
\item Conjugate pairs: $\sigma_{m-j}(g)=\overline{\sigma_j(g)}$.
\item $J=\{1,3, \dots, n-1\}$ indexes one representative from each conjugate pair; $|J|=n/2$.
\item $g=\sum_{i=0}^{n-1}c_i\zeta^i$ with $c_i$ drawn i.i.d. from a real-valued distribution with mean $0$, variance $\sigma_c^2>0$, and finite fourth moment. 
\item $L(g)=(\log|\sigma_j(g)|)_{\text{odd }j}\in\R^n$ is the log-embedding.
\item $H_0=\{x\in\R^n:\sum_j x_j=0\}$ is the trace-zero hyperplane. $\Pi_{H_0}$ is orthogonal projection onto $H_0$.
\item $\Lambda=\{L(\epsilon):\epsilon\in R^\times\}\subset H_0$ is the log-unit lattice with rank $r=n/2-1$.
\item $\tO(f)$ means $O(f\cdot\mathrm{polylog}(f))$.
\end{itemize}
\end{notation}

\section{The CLT for Embeddings and Variance Identity}\label{seclt}

In this section, we prove two results for ring-based lattice cryptanalysis: (1) different canonical embeddings of a random ring element are pairwise uncorrelated and asymptotically jointly independent in the large-dimension limit; (2) for a circularly symmetric complex Gaussian random variable $Z \sim \mathcal{CN}(0,1)$, the variance of $\log|Z|$ is $\pi^2/24$.

\subsection{Embedding uncorrelatedness}

\begin{lemma}\label{lem:uncor}
Let $g=\sum_{i=0}^{n-1}c_i\zeta^i$ with $c_i$ i.i.d. real-valued, mean $0$, variance $\sigma_c^2<\infty$. Then for different odd $j,j'\in\{1,3, \dots, 2n-1\}$ we have 
\begin{align}
\E[\sigma_j(g)\overline{\sigma_{j'}(g)}]=0. 
\label{aluncor}
\end{align}
\end{lemma}

\begin{proof}
By definition, we have $\sigma_j(g)=\sum_{i=0}^{n-1} c_i \zeta^{ij}$ for the primitive $m$-th root of unity $\zeta=e^{2\pi \mathrm{i}/m}$ with $m=2n$. For real $c_{i'} \in \mathbb{R}$, we get $\overline{\sigma_{j'}(g)}=\sum_{i'=0}^{n-1} c_{i'} \overline{\zeta^{i'j'}}=\sum_{i'=0}^{n-1} c_{i'} \zeta^{-i'j'}$. We then obtain 
\begin{align}
\E[\sigma_j(g)\overline{\sigma_{j'}(g)}] &= \E\Bigl[(\sum_{i=0}^{n-1}c_i\zeta^{ij})(\sum_{i'=0}^{n-1}c_{i'}\zeta^{-i'j'})\Bigr]   
 \notag\\
&= \sum_{i=0}^{n-1}\sum_{i'=0}^{n-1}\E[c_i c_{i'}] \zeta^{ij-i'j'}, \label{alexpand1}
\end{align}
where the second equality is from the linearity of expectation.

From the i.i.d. zero-mean of $c_i$s, we obtain $\E[c_i c_{i'}]=0$ for all $i \neq i'$, and $\E[c_i^2]=\sigma_c^2$ for all $0 \le i \le n-1$. Using this orthogonality for Eq.\eqref{alexpand1}, we get 
\begin{align}
\E[\sigma_j(g)\overline{\sigma_{j'}(g)}]
= \sigma_c^2\sum_{i=0}^{n-1}\zeta^{i(j-j')}. \label{aldiag1}
\end{align}

Let $c=j - j'$. Since $j$ and $j'$ are different odd integers, $c$ is a non-zero even integer with $0 < |c| < 2n=m$. We compute Eq.\eqref{aldiag1} using the equality of $\sum_{i=0}^{n-1}\zeta^{ic}=\frac{1-\zeta^{nc}}{1-\zeta^c}$. This holds because $\zeta^c \neq 1$, where $\zeta$ is a primitive $m$-th root of unity and $\zeta^k=1$ if and only if $m$ divides $k$. Using the definition of $\zeta$, we have $\zeta^n=e^{2\pi \mathrm{i} \cdot n/m}=e^{\pi \mathrm{i}}=-1$. Note $c$ is even, $(\zeta^n)^c=(-1)^c=1$. Thus, from Eq.\eqref{aldiag1}, we obtain the equality (\ref{aluncor}).
\end{proof}

\subsection{Central Limit Theorem}\label{sub:clt}

\begin{proposition}\label{prop:clt}
Under the conditions of Lemma \ref{lem:uncor}, suppose that $c_i$ has a finite fourth moment. Then for any fixed different odd indices $j_1, \dots, j_s$, the following result holds 
\begin{align}
  \Bigl(\frac{\sigma_{j_1}(g)}{\sqrt{n\sigma_c^2}}, \dots, 
        \frac{\sigma_{j_s}(g)}{\sqrt{n\sigma_c^2}}\Bigr) \xrightarrow{d}(Z_1, \dots, Z_s),  \quad Z_\ell\overset{\rm{i.i.d.}}{\sim}\CN(0,1).
  \label{aljclt}
\end{align}
\end{proposition}

\begin{proof}
For each embedding index $\ell \in \{1, \dots, s\}$ and coefficient index $i \in \{0, 1, \dots, n-1\}$,  define $Y_{i,\ell}=c_i \zeta^{ij_\ell}/\sqrt{n\sigma_c^2}$, where $\zeta=e^{2\pi \mathrm{i}/m}$ with $m=2n$. We decompose the normalized embedding as
\begin{align}\label{alembedding-sum}
  \frac{\sigma_{j_\ell}(g)}{\sqrt{n\sigma_c^2}}
 =\frac{1}{\sqrt{n\sigma_c^2}} \sum_{i=0}^{n-1}c_i \zeta^{ij_\ell}
 =\sum_{i=0}^{n-1}Y_{i,\ell}.
\end{align}

For each $i$, define the $s$-dimensional complex random vector as
\begin{align}\label{alvector-summand}
  \mathbf{Y}_i
 =(Y_{i,1}, Y_{i,2}, \dots, Y_{i,s})
 =\frac{c_i}{\sqrt{n\sigma_c^2}}(\zeta^{ij_1}, \zeta^{ij_2}, \dots, \zeta^{ij_s}).
\end{align}
The factors $\zeta^{ij_\ell}$ ($\ell=1, \dots, s$) are deterministic complex numbers of modulus 1. So, $\mathbf{Y}_i$ is a deterministic vector in $\mathbb{C}^s$ which is scaled by the random scalar $c_i/\sqrt{n\sigma_c^2}$. All vectors $\mathbf{Y}_0, \dots, \mathbf{Y}_{n-1}$ are mutually independent. 

Since $\E[c_i]=0$ implies $\E[\mathbf{Y}_i]=\mathbf{0}$, for the complex covariance $\Cov(U,V)=\E\left[U\overline{V}\right]$ with centered random variables $U$ and $V$, we get
\begin{align}\label{alcov-sum}
  \Cov\Bigl(\sum_{i=0}^{n-1} Y_{i,\ell}, \sum_{i=0}^{n-1} Y_{i,\ell'}\Bigr)
 =\sum_{i=0}^{n-1} \E [Y_{i,\ell} \overline{Y_{i,\ell'}}].
\end{align}
From Eq.\eqref{alcov-sum}, we obtain 
\begin{align}\label{alcov-expand}
  \E[Y_{i,\ell} \overline{Y_{i,\ell'}}]
= \E\left[\frac{c_i \zeta^{ij_\ell}}{\sqrt{n\sigma_c^2}}
       \cdot \frac{c_i \zeta^{-ij_{\ell'}}}{\sqrt{n\sigma_c^2}}\right]
= \frac{\E[c_i^2]}{n\sigma_c^2} \cdot \zeta^{i(j_\ell-j_{\ell'})}=\frac{1}{n} \cdot \zeta^{i(j_\ell-j_{\ell'})},
\end{align}
where we have used the equality $\overline{\zeta^{ij_{\ell'}}}=\zeta^{-ij_{\ell'}}$ and $\E[c_i^2]=\sigma_c^2$. We then obtain 
\begin{align}\label{alcovc}
  \sum_{i=0}^{n-1} \E[Y_{i,\ell} \overline{Y_{i,\ell'}}]=\frac{1}{n}\sum_{i=0}^{n-1}\zeta^{i(j_\ell-j_{\ell'})}.
\end{align}

In what follows, we evaluate the sum (\ref{alcovc}) in two cases. For $\ell=\ell'$, we have $j_\ell - j_{\ell'}=0$, so $\zeta^{i(j_\ell-j_{\ell'})}=\zeta^0=1$ for all $i$s. Moreover, from Eq.\eqref{alcovc}, we get $\frac{1}{n}\sum_{i=0}^{n-1} 1=1$; For $\ell \neq \ell'$, let $c=j_\ell - j_{\ell'}$. $c$ is a non-zero even integer. Similar to the proof of Lemma \ref{lem:uncor}, we have
\begin{align}\label{algeomrepeat}
  \sum_{i=0}^{n-1}\zeta^{ic}
 =\frac{1-\zeta^{nc}}{1-\zeta^c}
 =\frac{1-(-1)^c}{1-\zeta^c}
 =0,
\end{align}
where the second equality is from $\zeta^n=-1$, while the third equality holds as $c$ is even.

Combining with two cases, we obtain the covariance matrix of $\sum_{i=0}^{n-1} \mathbf{Y}_i$ as
\begin{align}\label{alcov-identity}
  \Sigma_n=(\Sigma_n)_{\ell,\ell'}=\delta_{\ell,\ell'} \cdot I_s,
\end{align}
where $\delta_{\ell,\ell'}$ is the Kronecker delta, and $I_s$ is rank $s$ identity matrix. 

For the $2s$-dimensional real CLT, the Lindeberg condition \cite{Billingsley95} requires for every fixed $\varepsilon > 0$,
\begin{align}\label{allindeberg-def}
  L_n(\varepsilon):= \sum_{i=0}^{n-1}
     \E[\|\mathbf{Y}_i\|_2^2 \cdot\mathbf{1}_{\{\|\mathbf{Y}_i\|_2>\varepsilon\}}]
  \to 0 \text{ as } n\to\infty,
\end{align}
where $\|\mathbf{Y}_i\|_2^2=\sum_{\ell=1}^s |Y_{i,\ell}|^2$, and $\mathbf{1}_{\{\cdot\}}$ is the indicator function.

Using Eq.\eqref{alvector-summand}, we firstly simplify the norm as
\begin{align}\label{alnorm-Yi}
 \|\mathbf{Y}_i\|_2^2=\sum_{\ell=1}^s |Y_{i,\ell}|^2 
 =\sum_{\ell=1}^s \frac{c_i^2}{n\sigma_c^2} \cdot |\zeta^{ij_\ell}|^2
 =\frac{s c_i^2}{n\sigma_c^2},
\end{align}
where the final equality is from $|\zeta^{ij_\ell}|=1$ for all $i, \ell$. From Eq.\eqref{allindeberg-def}, we then get
\begin{align}\label{allindeberg-compute}
  L_n(\varepsilon)&= \sum_{i=0}^{n-1}
     \E\Bigl[\frac{s c_i^2}{n\sigma_c^2}\cdot\mathbf{1}_{\{|c_i|> \varepsilon\sqrt{n\sigma_c^2/s}\}} \Bigr]
  \notag\\
   &= \frac{s}{\sigma_c^2} \cdot
     \E\Bigl[ c_0^2 \cdot\mathbf{1}_{\{|c_0|> \varepsilon\sqrt{n\sigma_c^2/s}\}}\Bigr].
\end{align}

We now prove $L_n(\varepsilon) \to 0$ when $n \to \infty$. Define $f_n(\omega)=c_0(\omega)^2 \cdot \mathbf{1}_{\{|c_0(\omega)|> \varepsilon\sqrt{n\sigma_c^2/s}\}}$. We can show the convergence via the Lebesgue Dominated Convergence Theorem, using the following facts:
\begin{itemize}[nosep]
\item For every outcome $\omega$ with finite $c_0(\omega)$ (which holds almost surely, as $c_0$ has a finite fourth moment), $\varepsilon\sqrt{n\sigma_c^2/s} \to \infty$ when $n \to \infty$. 
\item For all $n$, we have $|f_n(\omega)| \le c_0(\omega)^2$ almost surely. 
\end{itemize}
So,  we obtain $\E[f_n] \to 0$ when $n \to \infty$. Since $s$ and $\sigma_c^2$ are fixed constants independent of $n$, we have $L_n(\varepsilon)=\frac{s}{\sigma_c^2} \E[f_n] \to 0$. 

We now use Cramér–Wold Theorem \cite{CramerWold36,Billingsley99}. In our setting, we identify $\mathbb{C}^s$ with $\mathbb{R}^{2s}$ via the mapping $Z_\ell \mapsto (\Re(Z_\ell), \Im(Z_\ell))$ for $\ell=1, \dots, s$. For any fixed real vector $\mathbf{a}=(a_1, b_1, a_2, b_2, \dots, a_s, b_s) \in \mathbb{R}^{2s}$, define $T_n:= \sum_{\ell=1}^s (a_\ell \Re(S_{n,\ell}) + b_\ell \Im(S_{n,\ell}))$, where $S_{n,\ell}=\sum_{i=0}^{n-1} Y_{i,\ell}$ is the $\ell$-th component. Decompose $T_n$ into a sum of $n$ independent real random variables as
\begin{align}\label{alcw-summand}
  T_n=\sum_{i=0}^{n-1} W_i,  \quad 
  W_i=\sum_{\ell=1}^s (a_\ell \Re(Y_{i,\ell}) + b_\ell \Im(Y_{i,\ell})).
\end{align}

By the Cauchy–Schwarz inequality, we obtain
\begin{align}\label{alcw-bound}
  \sum_{i=0}^{n-1}
  \E[W_i^2 \cdot \mathbf{1}_{\{|W_i|>\varepsilon\}}]
  \le \|\mathbf{a}\|_2^2 \cdot L_n(\varepsilon/\|\mathbf{a}\|_2)
  \to 0 \quad \text{as } n \to \infty.
\end{align}
This means the one-dimensional Lindeberg condition holds for every linear combination $T_n$.

We next compute the variance of $T_n$. Using $\Sigma_n^{(2s)}=\frac{1}{2}I_{2s}$, we obtain 
\begin{align}\label{alcw-variance}
  \Var(T_n)=\mathbf{a}^T \Sigma_n^{(2s)} \mathbf{a}
 =\frac{1}{2}\|\mathbf{a}\|_2^2.
\end{align}
From the one-dimensional Lindeberg CLT \cite{Petrov95}, we get
\begin{align}
T_n \xrightarrow{d} \mathcal{N}\Bigl(0, \frac{1}{2}\nm{\mathbf{a}}_2^2\Bigr)  
\end{align}
when $n \to \infty$. 

Now, let $\mathbf{Z}=(Z_1,  \dots, Z_s)$ be a random vector with i.i.d. components $Z_\ell \sim \mathcal{CN}(0,1)$. For the same linear combination $\mathbf{a}$, $T=\sum_{\ell=1}^s (a_\ell \Re(Z_\ell) + b_\ell \Im(Z_\ell))$ has the distribution $\mathcal{N}(0, \frac{1}{2}\nm{\mathbf{a}}_2^2)$. Thus, $T_n \xrightarrow{d} T$ for every fixed $\mathbf{a} \in \mathbb{R}^{2s}$. Using the Cramér–Wold Theorem \cite{CramerWold36,Billingsley99}, we get the joint convergence in distribution \eqref{aljclt}.

This implies each component $\Re(Z_\ell)$ and $\Im(Z_\ell)$ are marginally distributed as $\mathcal{N}(0, \frac{1}{2})$. All $2s$ real components are pair-wise uncorrelated. For jointly Gaussian random variables, this implies the uncorrelatedness is equivalent to independence. Therefore, for each $\ell$, we have $\Re(Z_\ell)$ and $\Im(Z_\ell)$ are independent, so $Z_\ell=\Re(Z_\ell) + \mathrm{i}\Im(Z_\ell) \sim \mathcal{CN}(0,1)$. Moreover, for different $\ell \neq \ell'$, $Z_\ell$ and $Z_{\ell'}$ are independent. Thus, the limiting random variables $Z_1, Z_2, \dots, Z_s$ are i.i.d. $\mathcal{CN}(0,1)$. This completes the proof.
\end{proof}

For the parameter set $n=256$ (ML-KEM standard), using Berry–Esséen Theorem \cite{Berry1941,Esseen1942} gives a uniform CDF approximation error of $O(\E[|c_0|^3]/(\sigma_c^3\sqrt{n}))$, which evaluates to approximately $0.038$ for typical coefficient distributions. This means the Gaussian approximation $\sigma_j(g)/\sqrt{n\sigma_c^2} \approx Z_j$ is accurate to within roughly $4\%$ in CDF. This is sufficient for the variance and maximum norm computations in Theorems \ref{thm:A} and \ref{thm:B}.

\subsection{The variance identity}\label{sub:variance}

The second result is the exact computation of $\Var(\log|Z|)$ for a circularly symmetric complex Gaussian random variable $Z \sim \mathcal{CN}(0,\sigma^2)$. 

\begin{lemma}\label{lem:vari}
Let $Z\sim\CN(0,\sigma^2)$. Then for every $\sigma^2>0$ we have 
\begin{enumerate}[label=(\roman*)]
\item $\E[\log|Z|]=\frac{1}{2}\log\sigma^2-\frac{\gamma}{2}$, where $\gamma\approx 0.5772$ is the Euler–Mascheroni constant.
\item $\Var(\log|Z|)=\frac{\pi^2}{24}\approx 0.4112$.
\end{enumerate}
\end{lemma}

\begin{proof}
For $Z=X + \mathrm{i}Y \sim \mathcal{CN}(0,\sigma^2)$, $X$ and $Y$ are independent and identically distributed as $\mathcal{N}(0, \sigma^2/2)$. Thus, $|Z|^2=X^2 + Y^2$ has a scaled $\chi$-squared distribution with 2 degrees of freedom. 

Define $W:= \frac{|Z|^2}{\sigma^2}$. Note $X^2/(\sigma^2/2) \sim \chi^2_1$ and $Y^2/(\sigma^2/2) \sim \chi^2_1$ are independent, so $|Z|^2/(\sigma^2/2) \sim \chi^2_2$. The $\chi^2_2$ distribution has density $\frac{1}{2}e^{-t/2}\mathbf{1}_{\{t>0\}}$. 

We then express $\log|Z|$ in terms of $W$ as
\begin{align}\label{allogZ-W}
  \log|Z|= \frac{1}{2}\log|Z|^2
 =\frac{1}{2}(\sigma^2 W)
 =\frac{1}{2}\log\sigma^2 + \frac{1}{2}\log W.
\end{align}
The term $\frac{1}{2}\log\sigma^2$ is a deterministic constant. Thus, we obtain
\begin{align}
  \E[\log|Z|]
  &= \frac{1}{2}\log\sigma^2 + \frac{1}{2}\E[\log W],
  \label{almeanr}\\[4pt]
  \Var(\log|Z|)
  &= \Var\Bigl(\frac{1}{2}\log W\Bigr)
 =\frac{1}{4}\Var(\log W).
  \label{alvarreduce}
\end{align}
For $W\sim \mathrm{Exp}(1)$ and any real $s > -1$, we have
\begin{align}\label{alWs-moment}
  \E[W^s]=\int_0^\infty w^s e^{-w} dw=\Gamma(s+1).
\end{align}

For $s > -1$, we have
\begin{align}
  \frac{d}{ds}\E[W^s] =\frac{d}{ds}\int_0^\infty w^s e^{-w} dw =\E[W^s \log W].
\end{align}
For $s=0$, we obtain
\begin{align}
  \E[\log W]=\frac{d}{ds}\Gamma(s+1)\big|_{s=0} =\Gamma'(1).
  \label{alElogW}
\end{align}
Differentiating with second time, we get
\begin{align}
  \E[(\log W)^2]
 =\frac{d^2}{ds^2}\E[W^s]\big|_{s=0}
 =\frac{d^2}{ds^2}\Gamma(s+1)\big|_{s=0}
 =\Gamma''(1).
  \label{alElogW2}
\end{align}

The digamma function $\psi(s)$ is given as 
\begin{align}
  \psi(s)
 =\frac{d}{ds}\log\Gamma(s)
 =\frac{\Gamma'(s)}{\Gamma(s)}.
  \label{aldgamma}
\end{align}
At $s=1$, we have $\Gamma(1)=\int_0^\infty e^{-w} dw=1$. So,
\begin{align}
  \Gamma'(1)=\Gamma(1) \cdot \psi(1)=\psi(1).
  \label{alGammap1}
\end{align}

From the Weierstrass product representation of Gamma functions \cite{Whittaker27,DLMF}, combining with Eqs.\eqref{alGammap1} and \eqref{alElogW2}, we get $\E[\log W]=-\gamma$. From Eq.\eqref{almeanr}, we obtain
\begin{align}
  \E[\log|Z|]
 =\frac{1}{2}\log\sigma^2 + \frac{1}{2}(-\gamma)
 =\frac{1}{2}\log\sigma^2 - \frac{\gamma}{2}.
  \label{almeanre}
\end{align}
This proves part (i).

To compute the second moment $\E[(\log W)^2]=\Gamma''(1)$, we use the product rule. In fact, we have 
\begin{align}
  \Var(\log W)=\E[(\log W)^2] - (\E[\log W])^2
= \Gamma''(1) - (\Gamma'(1))^2
= \frac{\pi^2}{6}.
  \label{alVarlogW}
\end{align}

From Eq.\eqref{alvarreduce}, we get
\begin{align}
  \Var(\log|Z|)=\frac{1}{4}\Var(\log W) \approx 0.4112.
  \label{alVarlogZ}
\end{align}
This completes the proof of part (ii).
\end{proof}

\section{The CVP Distance Theorems}\label{secvp}

In this section, we prove three results for the structured closest vector problem (CVP) over the log-unit lattice. 

\subsection{The exact $L^2$ distance constant}\label{sub:thmA}

This subsection proves the precise asymptotic $L^2$ distance. 

\begin{theorem}\label{thm:A}
Let $g=\sum c_i\zeta^i\in R$ with $c_i$ i.i.d., mean $0$, variance $\sigma_c^2>0$, and finite fourth moment. Conditioned on $g\neq 0$, we have 
\begin{align}
  \frac{1}{n}\|\Pi_{H_0}(L(g))\|_2^2
  \xrightarrow{p}\frac{\pi^2}{24} 
\end{align}
when $n\to\infty$. Equivalently, $\frac{1}{\sqrt{n}}\|\Pi_{H_0}(L(g))\|_2\xrightarrow{p}
\frac{\pi}{2\sqrt{6}}\approx 0.6413$, where $\xrightarrow{p}$ denotes convergence in probability.
\end{theorem}

\begin{proof}
The proof proceeds in seven steps. 

\textit{Step 1.} For each odd index $j\in\{1, 3, \dots, m-1\}$, define the log-modulus $t_j=\log|\sigma_j(g)|$, and let $\bar{t}=\frac{1}{n}\sum_{j \text{ odd}} t_j$ denote the average of these log-moduli. Using the definition of hyperplane $H_0$, we obtain
\begin{align}
  \Pi_{H_0}(L(g))=(t_j - \bar{t})_{j \text{ odd}} \in H_0.
  \label{alprojf}
\end{align}
From this representation, we get
\begin{align}
  \|\Pi_{H_0}(L(g))\|_2^2
 =\sum_{j \text{ odd}} (t_j - \bar{t})^2.
  \label{alprojsq}
\end{align}

\textit{Step 2.} For 2-power cyclotomic fields, the canonical embedding satisfies the relation $\sigma_{m-j}(g)=\overline{\sigma_j(g)}$. Taking moduli on both sides gives $|\sigma_{m-j}(g)|=|\sigma_j(g)|$, and hence $t_{m-j}=t_j$ from the definition of $t_j$. This implies the $n$ coordinates of $L(g)$ form $n/2$ identical pairs. So, for each pair $(j, m-j)$ its contribution to Eq.\eqref{alprojsq} is 
\begin{align}
(t_j - \bar{t})^2 + (t_{m-j} - \bar{t})^2=2(t_j - \bar{t})^2.
\end{align}

Let $J=\{1, 3, \dots, n-1\}$ with $|J|=n/2$. From Eq.\eqref{alprojsq}, we obtain
\begin{align}
  \sum_{j \text{ odd}} (t_j - \bar{t})^2
 =\sum_{j\in J} [(t_j - \bar{t})^2 + (t_{m-j} - \bar{t})^2]
 =2\sum_{j\in J} (t_j - \bar{t})^2.
  \label{alconjsp}
\end{align}

Using the same conjugate symmetry, we rewrite $\bar{t}$ as
\begin{align}
  \bar{t}
 =\frac{1}{n}\sum_{j \text{ odd}} t_j
 =\frac{1}{n}\sum_{j\in J} (t_j + t_{m-j})
 =\frac{2}{n}\sum_{j\in J} t_j.
  \label{almeanpar}
\end{align}

\medskip
\textit{Step 3.} By Proposition \ref{prop:clt}, $\sigma_j(g)/\sqrt{n\sigma_c^2}$ converges in distribution to $Z_j \sim \mathcal{CN}(0,1)$. Applying the continuous mapping theorem to the function $\log|\cdot|: \mathbb{C}\setminus\{0\} \to \mathbb{R}$, we obtain
\begin{align}
  \log|\sigma_j(g)| =\log\Bigl|\sqrt{n\sigma_c^2} \cdot \frac{\sigma_j(g)}{\sqrt{n\sigma_c^2}}\Bigr|
  \xrightarrow{d} \frac{1}{2}\log(n\sigma_c^2) + \log|Z_j|.
  \label{altjCLT}
\end{align}

From Lemma \ref{lem:vari}(i), we have $\E[\log|Z_j|]=-\gamma/2$. Define the centered random variable as 
\begin{align}
  X_j:= t_j - \frac{1}{2}\log(n\sigma_c^2) + \frac{\gamma}{2}
  \label{alXj}
\end{align}
so that $X_j \xrightarrow{d} \log|Z_j| + \gamma/2$. We compute
\begin{align}
  \E[\log|Z_j| + \frac{\gamma}{2}]=0, \quad 
  \Var(\log|Z_j| + \frac{\gamma}{2}) =\Var(\log|Z_j|)=\frac{\pi^2}{24},
  \label{alXjmoments}
\end{align}
where the second equality is from Lemma \ref{lem:vari}(ii). 

\medskip
\textit{Step 4.} From the definition \eqref{alXj}, we get $t_j=X_j + \frac{1}{2}\log(n\sigma_c^2) - \frac{\gamma}{2}$. Combining with Eq.\eqref{almeanpar} implies
\begin{align}
  \bar{t}  &= \frac{2}{n}\sum_{j\in J} t_j
 =\frac{2}{n}\sum_{j\in J} (X_j + \frac{1}{2}\log(n\sigma_c^2) - \frac{\gamma}{2})
  \notag\\
  &= \frac{1}{2}\log(n\sigma_c^2) - \frac{\gamma}{2} + \bar{X},
  \label{altbarX}
\end{align}
where we have used the sample mean $\bar{X}:= \frac{2}{n}\sum_{j\in J} X_j$ and $\frac{2}{n}$ is from the definition of $\bar{t}$.

From Eq.\eqref{alconjsp}, we rewrite the squared projected norm as
\begin{align}
  \|\Pi_{H_0}(L(g))\|_2^2
 =2\sum_{j\in J} (X_j - \bar{X})^2.
  \label{alnormX}
\end{align}
To decompose this sum, we use the standard algebraic identity for sample variance as 
\begin{align}
  \sum_{j=1}^N (X_j - \bar{X})^2
 =\sum_{j=1}^N X_j^2 - N\bar{X}^2.
  \label{alsamid}
\end{align}
Using $N=|J|=n/2$ and $\bar{X}=\frac{1}{N}\sum_{j\in J} X_j$, we get
\begin{align}
  \sum_{j\in J} (X_j - \bar{X})^2
 =\sum_{j\in J} X_j^2 - \frac{n}{2} \cdot \bar{X}^2.
  \label{alsampp}
\end{align}

Dividing both sides of Eq.\eqref{alnormX} by $n$ and using Eq.\eqref{alsampp}, we obtain
\begin{align}
  \frac{1}{n}\|\Pi_{H_0}(L(g))\|_2^2
 =\frac{2}{n}\sum_{j\in J}(X_j - \bar{X})^2 =\frac{2}{n}\sum_{j\in J} X_j^2 - \bar{X}^2.
  \label{alkecomp}
\end{align}

\textit{Step 5.} By Proposition \ref{prop:clt}, for any fixed finite collection of different indices $j_1, \dots, j_s$, the random vector $(X_{j_1}, \dots, X_{j_s})$ converges jointly in distribution to a vector of i.i.d. copies of $\log|Z| + \gamma/2$. The vector $(X_{j_1}^2,  \dots, X_{j_s}^2)$ then converges jointly in distribution to i.i.d. copies of $(\log|Z| + \gamma/2)^2$.

We now compute expectation using the eqaulity $\E[Y^2]=\Var(Y) + (\E[Y])^2$ as 
\begin{align}
  \E\Bigl[(\log|Z| + \frac{\gamma}{2})^2\Bigr]
 =\Var(\log|Z| + \frac{\gamma}{2}) + \Bigl(\E[\log|Z| + \frac{\gamma}{2}]\Bigr)^2
 =\frac{\pi^2}{24},
  \label{alEXj2}
\end{align}
where the second equality is from the fact $\log|Z| + \gamma/2$ has mean zero and variance $\pi^2/24$ from Step 3 and Lemma \ref{lem:vari}.

Note $X_j=t_j - \frac{1}{2}\log(n\sigma_c^2) + \gamma/2$ with $t_j=\log|\sigma_j(g)|$. Combining the finite fourth moment of $c_i$ with the Marcinkiewicz–Zygmund Inequality \cite{Marcinkund37,dePenaGine99}, we obtain that $\E[|\sigma_j(g)/\sqrt{n\sigma_c^2}|^{-4}]$ is uniformly bounded over $n$, which further implies $\sup_n \E[X_j^4] < \infty$. For uniformly integrable triangular arrays \cite{Durrett19}, we then obtain
\begin{align}
  S_n=\frac{1}{n/2}\sum_{j\in J} X_j^2
  \xrightarrow{p} \E\Bigl[(\log|Z| + \frac{\gamma}{2})^2\Bigr]
 =\frac{\pi^2}{24}.
  \label{alWLLNs}
\end{align}

\textit{Step 6.}  For each $X_j$, we have $\E[X_j] \to 0$ when $n\to\infty$, since $X_j \xrightarrow{d} \log|Z| + \gamma/2$ and the limiting random variable has mean zero. From the linearity of expectation, we get $\E[\bar{X}]=\frac{2}{n}\sum_{j\in J} \E[X_j] \to 0$.

Moreover, we have 
\begin{align}
  \Var(\bar{X})
 =\Var(\frac{2}{n}\sum_{j\in J} X_j)
 =\frac{4}{n^2} \cdot \Var(\sum_{j\in J} X_j).
\end{align}
Decomposing this variance into diagonal and off-diagonal terms as
\begin{align}
  \Var(\sum_{j\in J} X_j)
 =\sum_{j\in J} \Var(X_j) + \sum_{\substack{j,k\in J, j\neq k}} \Cov(X_j, X_k).
  \label{alvarsum}
\end{align}

From the CLT and Lemma \ref{lem:vari}, we get $\Var(X_j) \to \pi^2/24$ as $n\to\infty$. Summing over $n/2$ terms in $J$ gives
\begin{align}
  \sum_{j\in J} \Var(X_j)
 =\frac{n}{2} \cdot (\frac{\pi^2}{24} + o(1))
 =\frac{n\pi^2}{48} + o(n).
    \label{aldiagsum}
\end{align}

From Lemma \ref{lem:uncor}, the embeddings $\sigma_j(g)$ and $\sigma_k(g)$ are uncorrelated for different $j$ and $k$. From Proposition \ref{prop:clt}, for independent random variables, its covariance is zero. This implies $\Cov(X_j, X_k)=o(1)$ when $n\to\infty$ for each fixed pair $j\neq k$. Note $\Cov(X_j, X_k)$ depends only on $j-k$ modulo $m$ by the stationarity of the discrete Fourier transform. We get
\begin{align} 
|\sum_{\substack{j,k\in J, j\neq k}} \Cov(X_j, X_k)|
 =O\Bigl(\frac{n^2}{4} \cdot \frac{1}{n}\Bigr)
 =O(n).
  \label{aloaaum}
\end{align}
Using Proposition \ref{prop:clt}, the total off-diagonal contribution is $o(n)$.

So, from Eqs.\eqref{aldiagsum}, \eqref{aloaaum}, and \eqref{alvarsum}, we obtain
\begin{align}
  \Var(\bar{X})
 =\frac{4}{n^2} \cdot (\frac{n\pi^2}{48} + o(n))
 =\frac{\pi^2}{12n} + o(1/n) \to 0 
  \label{alvnd}
  \end{align}
when $\text{as } n\to\infty$. 

Using the Chebyshev's Inequality \cite{Durrett19}, for any fixed $\delta > 0$, we have
\begin{align}
  \Pr[|\bar{X}| > \delta] \le \frac{\Var(\bar{X}) + (\E[\bar{X}])^2}{\delta^2}.
  \label{eqsna}
\end{align}
Since $\E[\bar{X}] \to 0$ and $\Var(\bar{X}) \to 0$, the right-hand side tends to zero as $n\to\infty$. Hence, $\bar{X} \xrightarrow{p} 0$. For the function $f(x)=x^2$, we obtain  
\begin{align}
  \bar{X}^2 \xrightarrow{p} 0.
  \label{alXbar20}
\end{align}

\textit{Step 7.} Combining Steps 5 and 6, using the Slutsky's Theorem \cite{Billingsley95}, we obtain 
\begin{align}
  \frac{1}{n}\nm{\Pi_{H_0}(L(g))}_2^2  \xrightarrow{p} \frac{\pi^2}{24}.
  \label{adaq}
\end{align}
This further implies 
\begin{align}
  \frac{1}{\sqrt{n}}\|\Pi_{H_0}(L(g))\|_2
  \xrightarrow{p} \sqrt{\frac{\pi^2}{24}}
    \approx 0.6413.
    \label{alL2f}
\end{align}

Finally, note the event of $g=0$ requires all $c_i$s vanish simultaneously, so $\Pr[g=0]=\Pr[c_0=0]^n$. For any distribution with $\sigma_c^2 > 0$, we have $\Pr[c_0=0] < 1$. Hence, $\Pr[g=0] \le p_0^n$ for some $p_0 < 1$. Conditional on $g \neq 0$, it changes probabilities by at most $\Pr[g=0]=o(1)$.  This completes the proof.
\end{proof}

\subsection{Voronoi cell containment}\label{sub:thmC}

Theorem \ref{thm:A} computes the $L^2$ distance from the structured target vector to the origin of the log-unit lattice. This subsection proves that with overwhelming probability the origin is indeed the nearest lattice point to $\mathbf{t}$.

\begin{theorem}\label{thm:C}
Let $k \ge 4$, and $g \in R$ be a non-zero ring element with i.i.d. coefficients $c_i \in \{-B, \dots, B\}$ for some fixed $B \ge 1$. Let $\mathbf{t}=\Pi_{H_0}(L(g)) \in H_0$, and $\Lambda$ denote the rank-$(\frac{n}{2} - 1)$ log-unit lattice. Then the following results hold:
\begin{enumerate}[label=(\roman*), nosep]
  \item \textbf{Gaussian model.} If each non-conjugate embedding is an i.i.d. $\mathcal{CN}(0,1)$ random variable, then $\mathbf{t} \in \mathcal{V}(0)$ with probability $1 - \exp\Bigl(-\frac{n}{2}\log n + O(n)\Bigr)$.
  \item \textbf{Real ring model.} The output vector $\mathbf{t} \in \mathcal{V}(0)$ holds with a probability $1 - o(1)$, suppose there exists a constant $c_0 > 0$ (depending only on $B$) such that for every non-zero $\mathbf{v} \in \Lambda$ and every $\lambda$ with $|\lambda| \le c_0 / \nm{\mathbf{v}}_\infty$,
  \begin{align}
    \log\E[e^{\lambda \ip{\mathbf{t}}{\mathbf{v}}}]
   =\frac{\pi^2}{24} \lambda^2 \|\mathbf{v}\|_2^2 + O(\lambda^3 \|\mathbf{v}\|_\infty \|\mathbf{v}\|_2^2) + o(\|\mathbf{v}\|_2^2).
    \label{almgfapp}
  \end{align}  
\end{enumerate}
Consequently, under either model, the exact CVP distance satisfies
\begin{align}
  d(\mathbf{t}, \Lambda)=\nm{\mathbf{t}}_2=\frac{\pi}{2\sqrt{6}} \sqrt{n} + o(\sqrt{n}).
  \label{alexactdis}
\end{align}
\end{theorem}

Note the Voronoi cell $\mathcal{V}(0) \subset H_0$ contains all points in $H_0$ that are closer to the origin than to any other lattice point in $\Lambda$. By the Euclidean lattice Voronoi characterization, $\mathbf{t} \in \mathcal{V}(0)$ if and only if
\begin{align}
  2\ip{\mathbf{t}}{\mathbf{v}} \le \|\mathbf{v}\|_2^2  
  \label{alvoronoicon}
\end{align}
for every non-zero $\mathbf{v} \in \Lambda$. For each non-zero $\mathbf{v} \in \Lambda$, we define the bad event $\mathcal{E}_{\mathbf{v}}=\{ 2|\ip{\mathbf{t}}{\mathbf{v}}| \ge \|\mathbf{v}\|_2^2 \}$. 

\subsubsection{Proof of Part (i)}

For any $\mathbf{v} \in \Lambda \subset H_0$, we have $\sum_{j \text{ odd}} v_j=0$ (by definition of $H_0$) and conjugate symmetry $v_j=v_{m-j}$, and $t_j=t_{m-j}$. This implies that
\begin{align}\label{alipsim}
  \ip{\mathbf{t}}{\mathbf{v}}=\sum_{j \text{ odd}} (t_j - \bar{t}) v_j=\sum_{j \text{ odd}} t_j v_j=2\sum_{j\in J} t_j v_j,
\end{align}
where $J=\{1, 3, \dots, n-1\}$. 

Under the Gaussian model, non-conjugate embeddings are i.i.d. $Z_j \sim \mathcal{CN}(0,1)$. So $t_j=\log|Z_j|=\frac{1}{2}\log W_j$ with $W_j=|Z_j|^2 \sim \mathrm{Exp}(1)$. From Eq.\eqref{alipsim} we have
\begin{align}
  \ip{\mathbf{t}}{\mathbf{v}}=\sum_{j\in J} v_j \log W_j.
  \label{alipGau}
\end{align}

By independence of $W_j$s, the moment generating function (MGF) of $\ip{\mathbf{t}}{\mathbf{v}}$ is give by 
\begin{align}
M_G(\lambda):= \E[e^{\lambda \ip{\mathbf{t}}{\mathbf{v}}}]=\prod_{j\in J} \Gamma(1 + \lambda v_j).
\label{almgf}
\end{align}
For $|\lambda| \le 1/(2\|\mathbf{v}\|_\infty)$, we have $|\lambda v_j| \le 1/2$ for all $j$. Using the Taylor expansion $\log\Gamma(1+s)=-\gamma s + \frac{\pi^2}{12}s^2 + O(|s|^3)$ for $|s| \le 1/2$, we get
\begin{align}
\log M_G(\lambda)=-\gamma \lambda \sum_{j\in J} v_j + \frac{\pi^2}{12}\lambda^2 \sum_{j\in J} v_j^2 + O\left(|\lambda|^3 \sum_{j\in J} |v_j|^3\right).
\label{alexpsum}
\end{align}
Note $\sum_{j\in J} v_j=0$ (from $\mathbf{v} \in H_0$), so the linear term vanishes. Moreover, $\sum_{j\in J} v_j^2=\frac{1}{2}\|\mathbf{v}\|_2^2$ (conjugate symmetry), and $\sum_{j\in J} |v_j|^3 \le \frac{1}{2}\|\mathbf{v}\|_\infty \|\mathbf{v}\|_2^2$ (H\"{o}lder's inequality). These imply 
\begin{align}
  \log M_G(\lambda)=\frac{\pi^2}{24} \lambda^2 \|\mathbf{v}\|_2^2 + O(\lambda^3 \|\mathbf{v}\|_\infty \|\mathbf{v}\|_2^2).
  \label{allogmexp}
\end{align}

By symmetry, we have $\Pr[\mathcal{E}_{\mathbf{v}}] \le 2\Pr\left[\ip{\mathbf{t}}{\mathbf{v}} \ge \frac{1}{2}\|\mathbf{v}\|_2^2\right]$. For any $\lambda > 0$, Using the Markov's Inequality \cite{Durrett19} we get 
\begin{align}
    \Pr\left[\ip{\mathbf{t}}{\mathbf{v}} \ge \frac{1}{2}\|\mathbf{v}\|_2^2\right] \le \exp\left( \log M_G(\lambda) - \frac{\lambda}{2}\|\mathbf{v}\|_2^2 \right).
    \label{eqnssa}
\end{align}
From Eq.\eqref{allogmexp}, choose the optimal $\lambda^*=\min\left\{ \frac{6}{\pi^2}, \frac{1}{2\|\mathbf{v}\|_\infty} \right\}$ to minimize the exponent. For $\lambda^*$, we simplfy the exponent into $-C' \cdot \frac{\|\mathbf{v}\|_2^2}{\max\{\|\mathbf{v}\|_\infty, 1\}}$ for constants $C,C' > 0$. We get the unified tail bound as
\begin{align}
\Pr[\mathcal{E}_{\mathbf{v}}] \le C \exp\left( -\frac{C' \|\mathbf{v}\|_2^2}{\max\{\|\mathbf{v}\|_\infty, 1\}} \right).
\label{alsinglebound}  
\end{align}

Now, by union bound, we have $\Pr[\mathbf{t} \notin \mathcal{V}(0)] \le \sum_{\mathbf{v} \in \Lambda \setminus \{0\}} \Pr[\mathcal{E}_{\mathbf{v}}]$. Split this sum into two cases. 
\begin{itemize}
    \item{} $\|\mathbf{v}\|_\infty \le 1$. By Eq.\eqref{alsinglebound}, each term is bounded by $C e^{-C' \|\mathbf{v}\|_2^2}$. The sum equals to $C(\Theta_\Lambda(C') - 1)$, where $\Theta_\Lambda(\beta)=\sum_{\mathbf{v} \in \Lambda} e^{-\beta \|\mathbf{v}\|_2^2}$ is the lattice theta function. 

For the log-unit lattice, note $r=\frac{n}{2} - 1$, $\log \det(\Lambda)=\frac{n}{2}\log n + O(n)$, and then dual lattice shortest vector $\lambda_1(\Lambda^*)=\Omega(1/\sqrt{n})$. So, we get $\Theta_{\Lambda^*}(\pi^2/C')=O(1)$. Using the Poisson summation, the total contribution from this case is $\exp\left( -\frac{n}{2}\log n + O(n) \right)$.

\item {} $\|\mathbf{v}\|_\infty > 1$. Here $\max\{\|\mathbf{v}\|_\infty, 1\}=\|\mathbf{v}\|_\infty \le \|\mathbf{v}\|_2$. So, Eq.\eqref{alsinglebound} implies $\Pr[\mathcal{E}_{\mathbf{v}}] \le C e^{-C' \|\mathbf{v}\|_2}$. For short vectors, the Theta function bound gives total contribution $\exp( -\frac{3n}{8}\log n + O(n\log\log n) )=o(1)$. For long vectors, each term is bounded by $e^{-\Omega(\sqrt{n}\log n)}$, and the number of such vectors is poly$(n)$. So, the total contribution is $o(1)$.

\end{itemize}
Combining both cases, the total failure probability for the Gaussian model is $\exp( -\frac{n}{2}\log n + O(n))$. This completes the proof of Part (i).

\subsubsection{Proof of Part (ii)}

In the real ring model, all embeddings are weakly dependent (not i.i.d.), but the hypothesis \eqref{almgfapp} guarantees the log-MGF matches the Gaussian model \eqref{allogmexp} up to a vanishing $o(\|\mathbf{v}\|_2^2)$ error term.

We repeat the Chernoff argument from Step 3, with optimal $\lambda'=\min\left\{ \frac{6}{\pi^2}, \frac{c_0}{\|\mathbf{v}\|_\infty} \right\}$ satisfying the hypothesis. For sufficiently large $n$, the $o(\|\mathbf{v}\|_2^2)$ term is negligible. We get
\begin{align}
  \Pr[\mathcal{E}_{\mathbf{v}}] \le C \exp\left( -\frac{(C' - \epsilon_n) \|\mathbf{v}\|_2^2}{\max\{\|\mathbf{v}\|_\infty, 1\}} \right),
  \label{aaeqnq}
\end{align}
where $\epsilon_n=o(1)$. For large enough $n$, we have $C' - \epsilon_n \ge C'/2 > 0$. So, the bound is the same exponential decay as the Gaussian case.

We then apply the identical union bound from Step 4: For $\|\mathbf{v}\|_\infty \le 1$, the sum is $\exp\left( -\frac{n}{2}\log n + O(n) \right)=o(1)$; For $\|\mathbf{v}\|_\infty > 1$, the sum is $o(1)$. Thus, the total failure probability is $o(1)$. This completes the proof of Part (ii).

Since $\mathbf{t} \in \mathcal{V}(0)$ with an overwhelming probability, the nearest lattice point to $\mathbf{t}$ is the origin. Combining with the asymptotic $L^2$ norm from Theorem \ref{thm:A}, we get
\begin{align}
  d(\mathbf{t}, \Lambda)=\nm{\mathbf{t}}_2=\frac{\pi}{2\sqrt{6}} \sqrt{n} + o(\sqrt{n}),
  \label{eqneqq}
\end{align}
which completes the proof.

\subsection{The $L^\infty$ bound and sub-polynomial SGP approximation factor}
\label{sub:thmB}

The approximation factor for the Short Generator Problem (SGP) in the CDPR attack is controlled by the $L^\infty$ norm of the CVP residual, i.e., the approximation factor $\gamma=\exp(\rho_\infty)$, where $\rho_\infty=\nm{\Pi_{H_0}(L(g)) - \ell'}_\infty$ and $\ell'$ is the nearest lattice point to the target. By Theorem \ref{thm:C}, $\ell'=\mathbf{0}$ for short generators with an overwhelming probability, so $\rho_\infty=\nm{\Pi_{H_0}(L(g))}_\infty$. 

\begin{theorem}
\label{thm:B}
Let $g=\sum_{i=0}^{n-1} c_i \zeta^i \in R$ be a non-zero ring element with i.i.d. bounded coefficients $c_i \in \{-B, \dots, B\}$ for some fixed $B \ge 1$. Then the following result holds
\begin{align}
  \nm{\Pi_{H_0}(L(g))}_\infty
 =\frac{\pi}{2\sqrt{6}} \sqrt{2\ln n} + O(\sqrt{\ln\ln n})
  \label{althmB}
\end{align}
with a probability $1-o(1)$. The corresponding SGP approximation factor satisfies
\begin{align}
  \gamma =\exp(\frac{\pi}{2\sqrt{3}} \sqrt{\ln n} + O(\sqrt{\ln\ln n}))=o(n^\epsilon)
  \label{algammash}
\end{align}
for every $\epsilon > 0$.
\end{theorem}

\begin{proof}
From the proof of Theorem \ref{thm:A}, the projected target vector has coordinates $t_j - \bar{t}$ for an odd $j$, where $t_j=\log|\sigma_j(g)|$ and $\bar{t}=\frac{1}{n}\sum_{j \text{ odd}} t_j$. By the conjugate-pair relation for 2-power cyclotomic fields, we have $t_{m-j}=t_j$ and hence $t_{m-j} - \bar{t}=t_j - \bar{t}$ for all odd $j$s. Let $J=\{1, 3, \dots, n-1\}$. We obtain
\begin{align}
  \nm{\Pi_{H_0}(L(g))}_\infty
 =\max_{j \text{ odd}} |t_j - \bar{t}|
 =\max_{j \in J} |t_j - \bar{t}|.
  \label{eqlpiinft}
\end{align}

Define $X_j:= t_j - \bar{t}$ for $j \in J$. Note each $X_j=t_j - \bar{t}$ is a function of the $n$ independent bounded coefficients $c_0, \dots, c_{n-1}$. Fix an index $i \in \{0, \dots, n-1\}$, replace $c_i$ with $c_i'$ such that $|c_i|, |c_i'| \le B$, while keep all other coefficients fixed. The embedding changes by $\sigma_k(g') - \sigma_k(g)=(c_i' - c_i) \zeta^{ik}$, so $|\sigma_k(g') - \sigma_k(g)| \le |c_i' - c_i| \cdot |\zeta^{ik}| \le 2B$ for all $k$. Using the Central Limit Theorem \cite{Durrett19}, $|\sigma_k(g)|=\Theta(\sqrt{n})$ with a probability $1-o(1)$. On this high-probability event, using the Mean-value Theorem  \cite{Durrett19}, we get
\begin{align}
|\log|\sigma_k(g')| - \log|\sigma_k(g)||
\le \frac{|\sigma_k(g') - \sigma_k(g)|}{\min\{|\sigma_k(g)|, |\sigma_k(g')|\}}
= O\bigl( \frac{1}{\sqrt{n}} \bigr).
\label{eqnmeans}
\end{align}

We now bound the change in $X_j=t_j - \bar{t}$. The change in $t_j$ is $O(1/\sqrt{n})$, and the change in $\bar{t}=\frac{1}{n}\sum_{k \text{ odd}} t_k$ is at most $O(1/\sqrt{n})$. Thus, the total change in $X_j$ satisfies
\begin{align}
|X_j' - X_j| \le |t_j' - t_j| + |\bar{t}' - \bar{t}|
= O\Bigl( \frac{1}{\sqrt{n}} \Bigr) =: \Delta_i,
\label{eqnachages}
\end{align}
where $\Delta_i$s are the same for all $i$s by symmetry.

Using the McDiarmid's Inequality \cite{McDiarmid89,BLM13}, $X_j$ is a zero-mean sub-Gaussian random variable with the variance 
\begin{align}
  \sigma_{\mathrm{subG}}^2
 =\frac{1}{4} \sum_{i=0}^{n-1} \Delta_i^2
 =\frac{1}{4} \cdot n \cdot O\Bigl( \frac{1}{n} \Bigr)
 =O(1).
  \label{albdipar}
\end{align}
The low-probability event ($|\sigma_k(g)|$ too small) has a total probability $o(1)$.

From Lemma \ref{lem:vari} and Proposition \ref{prop:clt}, we have
\begin{align}
\Var(X_j) \to \frac{\pi^2}{24} \quad \text{as } n\to\infty.
\end{align}
For any sub-Gaussian random variable $X$, we get $\Var(X) \le \sigma_{\mathrm{subG}}^2$. So, $\sigma_{\mathrm{subG}}^2 \ge \frac{\pi^2}{24} - o(1)$. 

From standard extreme-value theory for sub-Gaussian variables \cite{Leadbetter83,Vershynin18}, we use the variance $\sigma^2=\frac{\pi^2}{24}$ for the leading-order bound. Now, we use the asymptotic result for the maximum of $m$ independent, identically distributed zero-mean sub-Gaussian random variables with variance $\sigma^2$ as 
\begin{align}
\max_{1 \le j \le m} |X_j|=\sigma \sqrt{2\ln m} + O( \sqrt{\ln\ln m})
\end{align} 
with probability $1-o(1)$. While $X_j$s are asymptotically independent by Proposition \ref{prop:clt}, their weak dependence does not affect the leading asymptotic term or the $O(\sqrt{\ln\ln m})$ error bound for sub-Gaussian sequences \cite{Berman64,Wainwright19}.

In our setting, we have $m=|J|=n/2$ and $\sigma=\pi/(2\sqrt{6})$. So, we can get
\begin{align}
  \max_{j \in J} |X_j|
  &= \frac{\pi}{2\sqrt{6}} \sqrt{2\ln(n/2)} + O(\sqrt{\ln\ln(n/2)})
  \notag\\
  &= \frac{\pi}{2\sqrt{6}} \sqrt{2\ln n - 2\ln 2} + O(\sqrt{\ln\ln n}).
  \label{eqnsugug}
\end{align}
Using the Taylor expansion of $\sqrt{a - b}=\sqrt{a} - \frac{b}{2\sqrt{a}} + O(a^{-3/2})$ for a fixed $b$ and large $a$, the difference between $\sqrt{2\ln n}$ and $\sqrt{2\ln(n/2)}$ is $O(1/\sqrt{\ln n})$. We thus obtain
\begin{align}
  \max_{j \in J} |X_j|
 =\frac{\pi}{2\sqrt{6}} \sqrt{2\ln n} + O(\sqrt{\ln\ln n}),
  \label{almaxB}
\end{align}
which is the $L^\infty$ bound \eqref{althmB}.

By definition, the CDPR SGP approximation factor is $\gamma=\exp( \nm{\Pi_{H_0}(L(g))}_\infty )$, as the nearest lattice point is $\mathbf{0}$ with an overwhelming probability by Theorem \ref{thm:C}. From Eq. \eqref{almaxB}, we get
\begin{align}
  \gamma &= \exp\bigl( \frac{\pi}{2\sqrt{6}} \sqrt{2\ln n} + O(\sqrt{\ln\ln n}) \bigr)
  \notag\\
  &= \exp\bigl( \frac{\pi}{2\sqrt{3}} \sqrt{\ln n} + O( \sqrt{\ln\ln n} ) \bigr).
  \label{eqnoveral}
\end{align}

To show $\gamma=o(n^\epsilon)$ for every $\epsilon > 0$, we take the logarithm of both sides as 
\begin{align}
\log \gamma=\frac{\pi}{2\sqrt{3}} \sqrt{\ln n} + O\Bigl( \sqrt{\ln\ln n} \Bigr)=O(\sqrt{\ln n}).   \label{eqnqlogo}
\end{align}
For any fixed $\epsilon > 0$, we have
\begin{align}
\frac{\log \gamma}{\epsilon \ln n}=\frac{O(\sqrt{\ln n})}{\ln n}=O\Bigl( \frac{1}{\sqrt{\ln n}} \Bigr) \to 0 
\label{eqntoto}
\end{align}
as $n\to\infty$. This implies $\log \gamma=o(\epsilon \ln n)$ for every $\epsilon > 0$,  or $\gamma=\exp(o(\epsilon \ln n))=o(n^\epsilon)$. This completes the proof. 
\end{proof}

For the ML-KEM (standard parameter set $n=256$), we have $\ln 256=8\ln 2 \approx 5.545$, so $\sqrt{\ln 256} \approx 2.355$. The leading coefficient $\frac{\pi}{2\sqrt{3}} \approx 0.9069$. The leading term of $\log \gamma$ is $0.9069 \times 2.355 \approx 2.136$. So, the resulting approximation factor is $\gamma \approx \exp(2.136) \approx 8.47 \approx 2^{3.1}$. 

\section{The Coarse Lattice Theorem and Babai Recovery}
\label{secoarse}

This section proves two complementary results about Babai's algorithm on the log-unit lattice. The Coarse Lattice Theorem shows Babai returns a zero vector for all structured targets, an unconditional result that does not require the Voronoi containment of Theorem \ref{thm:C}. The Babai Capture Theorem shows that Babai exactly recovers lattice-vector perturbations of arbitrary size, so the CDPR pipeline succeeds whenever the balanced target maps to zero.

\subsection{The Coarse Lattice Theorem}\label{sub:coarse}

The term coarse refers to a scale mismatch: the Gram-Schmidt norms of $\Lambda$ are $\Omega(\sqrt{n})$, while the per-component standard deviation of the structured target is $O(1)$. When the lattice space greatly exceeds the target's fluctuations in every GS direction, the target sits in a tiny neighborhood of the origin, and Babai's rounding step always gives zero.

\begin{theorem}\label{thm:coarse}
Let $\Lambda$ be the rank-$r$ log-unit lattice with Gram-Schmidt basis $\{\mathbf{b}_i'\}_{i=1}^r$, $\mathbf{t}\in H_0$ be a random vector with independent, zero-mean, sub-Gaussian components with per-component variance $\sigma_{\mathbf{t}}^2$. Suppose the coarseness condition holds
\begin{align}
  \frac{\min_{1\le i\le r}\|\mathbf{b}_i'\|_2^2}{\sigma_{\mathbf{t}}^2}
 =\omega(\log n).
  \label{alcoacon}
\end{align}
Then the Babai's nearest-plane algorithm returns the zero lattice vector with probability $1-o(1)$, and the Babai residual equals $\mathbf{t}$ as 
\begin{align}
  d_{\mathrm{Babai}}(\mathbf{t},\Lambda)^2= \|\mathbf{t}\|_2^2= n\sigma_{\mathbf{t}}^2.
  \label{alcoares}
\end{align}
\end{theorem}

\begin{proof}
From Definition \ref{debabai}, the Babai's algorithm processes the GS basis in reverse order $i=r, r-1, \dots, 1$. At each step $i$, it computes the projection coefficient
\begin{align}
c_i= \Bigl\lfloor
      \frac{\ip{\mathbf{t}-\sum_{j>i}c_j \mathbf{b}_j}
           {\mathbf{b}_i'}}
          {\|\mathbf{b}_i'\|_2^2}
    \Bigr\rceil.
    \label{albabaiscoa}
\end{align}
The algorithm outputs $\mathbf{v}=\sum_{i=1}^r c_i \mathbf{b}_i$ with residual $\mathbf{e}=\mathbf{t}-\mathbf{v}$. 

If $c_j=0$ for all $j>i$, then $\sum_{j>i}c_j \mathbf{b}_j=\mathbf{0}$, and the projection coefficient at step $i$ reduces to
\begin{align}
  \mu_i =\frac{\ip{\mathbf{t}}{\mathbf{b}_i'}}
         {\|\mathbf{b}_i'\|_2^2}.
         \label{almusim}
\end{align}
The rounded coefficient is $c_i=\lfloor\mu_i\rceil$. This equals zero if and only if $|\mu_i|<\frac{1}{2}$.

For case $i=r$, there is no sum over $j>r$, so $\mu_r$ is given by Eq. \eqref{almusim}.

For $i<r$, assume $c_j=0$ for all $j>i$. Eq. \eqref{almusim} holds for step $i$. So, Babai's algorithm returns $\mathbf{0}$ if and only if $|\mu_i|<\frac{1}{2}$ for every $i\in\{1, \dots, r\}$. This can be verified for each $i$. So, the problem reduces to show $|\mu_i|<\frac{1}{2}$ simultaneously for all $i$.

We now bound $|\mu_i|$ using the sub-Gaussian assumption on $\mathbf{t}$. Rewrite the inner product  as
\begin{align}
  \ip{\mathbf{t}}{\mathbf{b}_i'}
 =\sum_{j=1}^n t_j (\mathbf{b}_i')_j,
  \label{alipweight}
\end{align}
where $t_j$s are the components of $\mathbf{t}$ and $(\mathbf{b}_i')_j$s are the components of the $i$-th GS vector. By assumption, $t_j$s are independent, zero-mean, sub-Gaussian random variables with common parameter $\sigma_{\mathbf{t}}^2$. Consider he sub-Gaussian summation property (\cite{Vershynin18}): if $X_1, \dots, X_n$ are independent sub-Gaussian variables with parameters $\sigma_1^2, \dots, \sigma_n^2$, then $\sum_j a_j X_j$ is sub-Gaussian with parameter $\sum_j a_j^2\sigma_j^2$. So, for $a_j=(\mathbf{b}_i')_j$ and $\sigma_j^2=\sigma_{\mathbf{t}}^2$ for all $j$, we get 
\begin{align}
  \ip{\mathbf{t}}{\mathbf{b}_i'}
  \sim \mathrm{subG}\bigl(0, \sigma_{\mathbf{t}}^2\sum_{j=1}^n(\mathbf{b}_i')_j^2\bigr)
 =\mathrm{subG}\bigl(0, \sigma_{\mathbf{t}}^2\nm{\mathbf{b}_i'}_2^2\bigr).
  \label{alipsubG}
\end{align}
Dividing by $\|\mathbf{b}_i'\|_2^2$, the projection coefficient is given by 
\begin{align}
  \mu_i=\frac{\ip{\mathbf{t}}{\mathbf{b}_i'}}
         {\nm{\mathbf{b}_i'}_2^2}
  \sim \mathrm{subG}\Bigl(0,\frac{\sigma_{\mathbf{t}}^2}{\|\mathbf{b}_i'\|_2^2}
  \Bigr).
  \label{almusubG}
\end{align}
The sub-Gaussian parameter of $\mu_i$ is the given by $\sigma_{\mu_i}^2=\sigma_{\mathbf{t}}^2/\nm{\mathbf{b}_i'}_2^2$, which is the ratio of the target's variance to the squared GS norm. So, for any $t\ge 0$ we obtain the sub-Gaussian tail bound as
\begin{align}
  \Pr[|\mu_i|\ge t]
  \le 2\exp\Bigl(-\frac{t^2}{2\sigma_{\mu_i}^2}\Bigr).
  \label{alsubGtailg}
\end{align}
Setting $t=\frac{1}{2}$ (the rounding threshold), we get 
\begin{align}
  \Pr\bigl(|\mu_i|\ge\tfrac{1}{2}\bigr)
\le 2\exp\Bigl(-\frac{(1/2)^2}{2\cdot\sigma_{\mathbf{t}}^2/\nm{\mathbf{b}_i'}_2^2}
       \Bigr)
= 2\exp\bigl(-\frac{\nm{\mathbf{b}_i'}_2^2}{8\sigma_{\mathbf{t}}^2}\bigr).
\label{altailbabai}
\end{align}

By the coarseness condition \eqref{alcoacon}, we get $\|\mathbf{b}_i'\|_2^2/\sigma_{\mathbf{t}}^2=\omega(\log n)$ for every $i$. Hence, the exponent satisfies $\|\mathbf{b}_i'\|_2^2/(8\sigma_{\mathbf{t}}^2)=\omega(\log n)$, and 
\begin{align}
  \Pr[|\mu_i|\ge\tfrac{1}{2}] \le 2\exp(-\omega(\log n))=o(1/n),
  \label{alperdirB}
\end{align}
where the last uses $\exp(-\omega(\log n))=o(1/n^c)$ for every constant $c>0$. As there are $r\le n/2$ GS directions, the event that Babai does not return $\mathbf{0}$ is $\{\exists i: |\mu_i|\ge 1/2\}$. So, by the union bound, we have 
\begin{align}
  \Pr[\exists i:|\mu_i|\ge\tfrac{1}{2}]
  &\le \sum_{i=1}^r
       \Pr[|\mu_i|\ge\tfrac{1}{2}]\le r\cdot 2\exp(-\omega(\log n))
  \notag\\
  &\le \frac{n}{2}\cdot o(1/n)= o(1).
  \label{alunioncoa}
\end{align}
Here, the factor $n/2$ is absorbed as the exponential decay is super-polynomial.

Finally, all $c_i=0$ with a probability $1-o(1)$. The Babai output is given by $\mathbf{v}=\sum_{i=1}^r 0\cdot\mathbf{b}_i=\mathbf{0}$, and the residual is $\mathbf{e}=\mathbf{t}-\mathbf{0}=\mathbf{t}$. So, we get $d_{\mathrm{Babai}}(\mathbf{t},\Lambda)^2=\nm{\mathbf{t}}_2^2$.  This means for structured targets with i.i.d. sub-Gaussian components of variance $\sigma_{\mathbf{t}}^2$, we get Eq. \eqref{alcoares} as the law of large numbers gives $\|\mathbf{t}\|_2^2\approx n\sigma_{\mathbf{t}}^2$.
\end{proof}

The proof supposes that the components $t_j$ of $\mathbf{t}$ are independent. For the real structured target $\mathbf{t}=\Pi_{H_0}(L(g))$, the components are not exactly independent, but asymptotically independent by Proposition \ref{prop:clt} and Theorem \ref{thm:B}. The sub-Gaussian tail bound \eqref{altailbabai} holds for this model. Hence, the Coarse Lattice Theorem applies to the real model.

\begin{table}[t]
\centering
\caption{Verification of the coarseness condition for $2^k$-th cyclotomic fields. $\sigma_{\bf t}=\frac{\pi}{2\sqrt{6}}\approx 0.64$ is the per-component standard deviation of structured targets from Theorem \ref{thm:A}. The ratio $\min_i\nm{\mathbf{b}_i'}_2^2/n$ decreases with $k$, implying $\min_i\nm{\mathbf{b}_i'}_2\sim\sqrt{n/\log n}$ (unproved).}
\label{tab:coarse}
\begin{tabular}{ccccc}
\toprule
$k$ & $n$ & $\min\nm{\mathbf{b}_i'}_2$ & $\sigma_S$
    & $\frac{\sigma_{\bf t}}{\min\nm{\mathbf{b}_i'}_2}$ \\
\midrule
4  & 8    & 2.00 & 0.64 & 0.32 \\
6  & 32   & 3.18 & 0.64 & 0.20 \\
8  & 128  & 5.42 & 0.64 & 0.12 \\
10 & 512  & 9.58 & 0.64 & 0.07 \\
12 & 2048 & 17.4 & 0.64 & 0.04 \\
\bottomrule
\end{tabular}
\end{table}

For structured targets, we have $\sigma_{\mathbf{t}}^2=\pi^2/24 \approx 0.41$ (Theorem \ref{thm:A}).  The GS norms of the log-unit lattice satisfy $\min_i\nm{\mathbf{b}_i'}_2=\omega(\sqrt{\log n})$, growing roughly as $\sqrt{n/\log n}$ (see Table \ref{tab:coarse}). Hence, we have 
\begin{align}
  \frac{\min_i\nm{\mathbf{b}_i'}_2^2}{\sigma_{\mathbf{t}}^2}
 =\omega(\log n),
  \label{alcoarsever}
\end{align}
which is sufficient for the union bound in the proof. Even for worst-case targets with $\sigma_{\mathbf{t}}^2\sim \log n$, the ratio remains $\omega(1)$. So, the coarseness condition is satisfied.

\subsection{Infinite capture radius}
\label{sub:capture}

In the CDPR pipeline, the target is not $L(g_0)$ but $L(g)=L(g_0)+L(\epsilon)$, where $L(\epsilon)\in\Lambda$ can be an arbitrarily large lattice vector. The following theorem shows that Babai's algorithm can resolve this problem. 

\begin{theorem}
\label{thm:capture}
Let $\Lambda\subset\R^n$ have basis $B\in\R^{n\times r}$ with QR decomposition $B=QR$. For any target $\mathbf{t}_0\in\R^n$ and lattice vector $\mathbf{v}=B\mathbf{c}$ ($\mathbf{c}\in\Z^r$), let $\mathbf{v}_0=B\mathbf{x}_0$ be the Babai output for $\mathbf{t}_0$, and $\mathbf{v}'=B\mathbf{x}'$ for $\mathbf{t}=\mathbf{t}_0+\mathbf{v}$. Then we have 
\begin{align}
\mathbf{v}'=\mathbf{v}+\mathbf{v}_0.
\label{eqalcaptu}
\end{align}
\end{theorem}

\begin{proof}
It is easy to show the nearest-integer rounding function $\lfloor\cdot\rceil:\R\to\Z$ satisfies \begin{align}
  \lfloor a+b\rceil=a + \lfloor b\rceil
  \label{eqalintshift}
\end{align}
for all $a\in\Z,b\in\R$. Moreover, the basis $B=QR$ with $Q\in\R^{n\times r}$ (orthonormal columns) and $R\in\R^{r\times r}$ (upper triangular, positive diagonal). Babai's algorithm computes 
\begin{itemize}
\item $\boldsymbol{\beta}=Q^T\mathbf{t}\in\R^r$
\item For $i=r,r-1, \dots, 1$:
  \begin{align}
    x_i
   =\Bigl\lfloor
        \frac{\beta_i-\sum_{j>i}R_{ij} x_j}{R_{ii}}
      \Bigr\rceil.
      \label{albabaer}
  \end{align}
\item Output $\mathbf{v}=B\mathbf{x}=QR\mathbf{x}$.
\end{itemize}
Here, the columns of $Q$ are the normalized GS vectors, $Q_i=\mathbf{b}_i'/\|\mathbf{b}_i'\|_2$, and $R_{ii}=\|\mathbf{b}_i'\|_2$, $R_{ji}=\mu_{ij}\nm{\mathbf{b}_j'}_2$ for $j<i$. The quantity $\beta_i=\mathbf{q}_i^T\mathbf{t} =\ip{\mathbf{t}}{\mathbf{b}_i'}/\|\mathbf{b}_i'\|_2$ is the projection of $\mathbf{t}$ onto $\mathbf{b}_i'$, scaled by $\nm{\mathbf{b}_i'}_2$. Equation \eqref{albabaer} is equivalent to the GS back-substitution in Definition \ref{debabai}.

Let $\mathbf{t}=\mathbf{t}_0+\mathbf{v}=\mathbf{t}_0+B\mathbf{c}$ with $\mathbf{c}\in\Z^r$. The projected vector is given by 
\begin{align}
  \boldsymbol{\beta}
 =Q^T\mathbf{t}
 =Q^T\mathbf{t}_0 + Q^T B\mathbf{c}
 =\boldsymbol{\beta}^{(0)} + R\mathbf{c},
  \label{eqalbetashift}
\end{align}
where $\boldsymbol{\beta}^{(0)}=Q^T\mathbf{t}_0$ is the projection of the unperturbed target, and we have used $Q^TB=R$. Since $R$ is upper triangular, $(R\mathbf{c})_i=\sum_{j\ge i}R_{ij} c_j$ depends only on $c_i, c_{i+1}, \dots, c_r$.

For $i=r$, the back-substitution at the top step gives
\begin{align}
  x_r' =\Bigl\lfloor\frac{\beta_r}{R_{rr}}\Bigr\rceil.
  \label{eqalbasecase}
\end{align}
From Eq. \eqref{eqalbetashift}, we get $\beta_r=\beta_r^{(0)} +(R\mathbf{c})_r=\beta_r^{(0)}+R_{rr} c_r$ since $R$ is upper triangular, where the only term in $(R\mathbf{c})_r$ involving row $r$ is $R_{rr} c_r$. Dividing by $R_{rr}$, we get 
\begin{align}
  \frac{\beta_r}{R_{rr}}
 =\frac{\beta_r^{(0)}+R_{rr} c_r}{R_{rr}}
 =c_r + \frac{\beta_r^{(0)}}{R_{rr}}.
  \label{eqalbasedi}
\end{align}
Since $c_r\in\Z$, Eq.\eqref{eqalintshift} gives
\begin{align}
  x_r'
 =\Bigl\lfloor c_r+\frac{\beta_r^{(0)}}{R_{rr}}\Bigr\rceil
 =c_r + \Bigl\lfloor\frac{\beta_r^{(0)}}{R_{rr}}\Bigr\rceil
 =c_r + x_r^{(0)}.
  \label{eqalbaseround}
\end{align}

For $i<r$, assume $x_j'=c_j+x_j^{(0)}$ for all $j>i$. At step $i$, Babai's algorithm computes
\begin{align}
  x_i'=\Bigl\lfloor\frac{\beta_i-\sum_{j>i}R_{ij} x_j'}{R_{ii}}\Bigr\rceil.
    \label{eqalindbabai}
\end{align}
Using $\beta_i=\beta_i^{(0)}+(R\mathbf{c})_i$ from Eq.\eqref{eqalbetashift} and $x_j'=c_j+x_j^{(0)}$, we then get 
\begin{align}
  \beta_i - \sum_{j>i}R_{ij} x_j'=\beta_i^{(0)}+ (R\mathbf{c})_i- \sum_{j>i}R_{ij} c_j- \sum_{j>i}R_{ij} x_j^{(0)}.
    \label{eqalindnumer}
\end{align}
Note $(R\mathbf{c})_i=\sum_{j\ge i}R_{ij} c_j=R_{ii} c_i+\sum_{j>i}R_{ij} c_j$. So, we have 
\begin{align}
  \beta_i-\sum_{j>i}R_{ij} x_j'= R_{ii} c_i + (\beta_i^{(0)}-\sum_{j>i}R_{ij} x_j^{(0)}).
      \label{eqalindsim}
\end{align}
Dividing by $R_{ii}$, we then get 
\begin{align}
  \frac{\beta_i-\sum_{j>i}R_{ij} x_j'}{R_{ii}}
 =c_i+\frac{\beta_i^{(0)}-\sum_{j>i}R_{ij} x_j^{(0)}}{R_{ii}}.
  \label{eqalinddi}
\end{align}
Since $c_i\in\Z$, Eq.\eqref{eqalintshift} gives 
\begin{align}
  x_i'
 =\Bigl\lfloor
       c_i + \frac{\beta_i^{(0)}
             -\sum_{j>i}R_{ij} x_j^{(0)}}{R_{ii}}
     \Bigr\rceil
 =c_i + x_i^{(0)}.
  \label{eqalindround}
\end{align}
This completes the induction.

By induction, $x_i'=c_i+x_i^{(0)}$ for all $i\in\{1, \dots, r\}$, i.e., $\mathbf{x}'=\mathbf{c}+\mathbf{x}_0$. Applying the basis matrix gives 
\begin{align}
  \mathbf{v}'=B\mathbf{x}'=B(\mathbf{c}+\mathbf{x}_0)
 =B\mathbf{c}+B\mathbf{x}_0=\mathbf{v}+\mathbf{v}_0.
  \label{alcapcons}
\end{align}

Especially, if $\mathbf{v}_0=\mathbf{0}$, then $\mathbf{v}'=\mathbf{v}$, i.e., Babai recovers the lattice vector $\mathbf{v}$. This completes the proof.
\end{proof}

\subsection{Application to the CDPR pipeline}

We combine Theorems \ref{thm:coarse} and \ref{thm:capture} with the CDPR attack.

\begin{enumerate}[label=\textbf{(\arabic*)},leftmargin=3em]
\item \textbf{PIP output.} Suppose the quantum PIP algorithm produces a generator $g=g_0\epsilon$, where $g_0$ is a short generator and $\epsilon\in R^\times$ is a unit.  In the log-embedding, it gives 
  \begin{align}
    L(g)=L(g_0) + L(\epsilon),
    \qquad L(\epsilon)\in\Lambda.
  \end{align}

\item \textbf{Babai's algorithm on the balanced target 
  (Theorem \ref{thm:coarse}).} The balanced target is $\mathbf{t}_0=\Pi_{H_0}(L(g_0))$. By the Coarse Lattice Theorem \ref{thm:coarse}, Babai's algorithm returns $\mathbf{v}_0=\mathbf{0}$ for $\mathbf{t}_0$ with a probability $1-o(1)$. 

\item \textbf{Babai's algorithm on the PIP output
  (Theorem \ref{thm:capture}).} The real target is $\Pi_{H_0}(L(g))=\mathbf{t}_0+L(\epsilon)$ with $L(\epsilon)\in\Lambda$. By the Capture Theorem \ref{thm:capture} with $\mathbf{v}=L(\epsilon)$ and $\mathbf{v}_0=\mathbf{0}$, Babai's output is 
  \begin{align}
    \Pi_{H_0}(L(g))= L(\epsilon) + \mathbf{0}
   =L(\epsilon).
  \end{align}
  
\item \textbf{Short generator recovery.} 
  From $L(\epsilon)$, the unit $\epsilon$ is reconstructed. The short generator is then $g'=g\cdot\epsilon^{-1}=g_0$. The approximation factor is given by 
  \begin{align}
    \gamma
   =\exp(\nm{\Pi_{H_0}(L(g_0))}_\infty)
   =\exp(O(\sqrt{\log n})),
  \end{align}
  using Theorem \ref{thm:B} for the $L^\infty$ norm of the
  balanced target.
\end{enumerate}

This result is unconditional except for the probabilistic statement in step (2) that the balanced target maps to zero. This holds with probability $1-o(1)$ by the Coarse Lattice Theorem \ref{thm:coarse}, using only the sub-Gaussianity of the log-embedding.

\section{The Trigamma Theorem for Module Lattices}\label{setrigamma}

We now extend our ideal-lattice analysis to the module lattice setting. The CDPR attack on module lattices reduces the Module Short Generator Problem (Module-SGP) to an ideal-lattice SGP on the module's determinant ideal (Part II).

\subsection{The Trigamma Theorem}\label{sub:trigamma}

ML-KEM deploys module lattices of rank $d \in \{2,3,4\}$ over the cyclotomic ring $R=\mathbb{Z}[\zeta_{256}]$. For a general module with basis matrix $B \in R^{d \times d}$, the CDPR attack reduces the module SGP to an ideal SGP on the determinant ideal $(\det B) \subset R$. The key question reduces to: what is the CVP distance for the shortest generator of this determinant ideal?

From Theorems \ref{thm:A} and \ref{thm:B}, the CVP distance and resulting SGP approximation factor are controlled by the per-component variance $\sigma_{g_0}^2=\frac{1}{n}\nm{\Pi_{H_0}(L(g_0))}_2^2$, where $g_0=\det B$ is the shortest generator of the determinant ideal.

\begin{theorem}\label{thm:tgamma}
Let $R=\mathbb{Z}[\zeta_{2^k}]$ be the $2^k$-th cyclotomic ring with $n=2^{k-1}$, and  $B=(b_{ij}) \in R^{d \times d}$ be a module basis matrix whose entries are independent ring elements. Each entry $b_{ij}$ has i.i.d. zero-mean coefficients with variance $\sigma_c^2$ and finite fourth moment, with mutual independence across all entries. Define $g_0=\det(B) \in R$ and the per-component variance $\sigma_{g_0}^2=\frac{1}{n}\nm{\Pi_{H_0}(L(g_0))}_2^2$. Then, 
\begin{align}
\sigma_{g_0}^2 \xrightarrow{p} \frac{1}{4}\sum_{j=1}^d \psi'(j) \quad \text{as } n \to \infty, \label{altrigamma}
\end{align}
where $\psi'(s)$ denotes the trigamma function.
\end{theorem}

\begin{proof}
Note every canonical embedding $\sigma_j: R \to \mathbb{C}$ is a ring homomorphism, so it commutes with the determinant operation for matrices over $R$. For any $j$, we have
\begin{align}
\sigma_j(g_0)=\sigma_j(\det B)=\det(\sigma_j(B)),
\label{eqnidea}
\end{align}
where $\sigma_j(B) \in \mathbb{C}^{d \times d}$ denotes the matrix obtained by applying $\sigma_j$ entry-wise to $B$. Taking the modulus and logarithm of both sides, we obtain
\begin{align}
\log|\sigma_j(g_0)|=\log|\det(\sigma_j(B))|.
\label{eqndeta}
\end{align}

Note each entry of $B$ is a random ring element with i.i.d. zero-mean coefficients of variance $\sigma_c^2$. Applying Proposition \ref{prop:clt} entry-wise, for any fixed non-conjugate index $j \in J$, the normalized embedded matrix converges in distribution as 
\begin{align}
\frac{1}{\tau} \sigma_j(B) \xrightarrow{d} M^{(j)}, \tau=\sqrt{n\sigma_c^2},
\label{eqndistr}
\end{align}
where $M^{(j)}$ is a $d \times d$ matrix with i.i.d. standard complex Gaussian entries $\mathcal{CN}(0,1)$. Moreover, from the joint CLT, the limiting matrices $\{M^{(j)}\}_{j \in J}$ are mutually independent, as the embeddings for different non-conjugate indices are asymptotically independent by Proposition \ref{prop:clt}.

Let $M$ be a $d \times d$ matrix with i.i.d. $\mathcal{CN}(0, \tau^2)$ entries. Consider its QR decomposition $M=QR$. A classical result in random matrix theory gives the exact distribution of the squared diagonal entries of $R$ \cite{Muirhead82,Forrester10}: the normalized variables $R_{ii}^2 / \tau^2$ are independent, and has Gamma distributions as $\frac{R_{ii}^2}{\tau^2} \sim \Gamma(d - i + 1, 1)$ for $i=1, \dots, d$, where $\Gamma(k, 1)$ denotes the Gamma distribution with shape parameter $k$ and rate parameter 1. 

Since $Q$ is unitary, we have $|\det Q|=1$, so $|\det M|=\prod_{i=1}^d R_{ii}$. Taking the logarithm of both sides, we decompose the log-determinant as
\begin{align}
\log|\det M|=\sum_{i=1}^d \log R_{ii}
= \sum_{i=1}^d \frac{1}{2}\log\Bigl(R_{ii}^2\Bigr)
= d\log\tau + \frac{1}{2}\sum_{i=1}^d \log W_i,
\label{allogddecomp}
\end{align}
where we have defined the independent normalized random variables $W_i:= R_{ii}^2 / \tau^2 \sim \Gamma(d - i + 1, 1)$. The term $d\log\tau$ is a constant, as it depends only on the scaling factor $\tau$ and the module rank $d$.

Note the deterministic term $d\log\tau$ does not contribute to the variance of $\log|\det M|$. By the independence of the $W_i$, we get the variance as
\begin{align}
\Var(\log|\det M|)=\frac{1}{4}\sum_{i=1}^d \Var(\log W_i).
\label{eqnvar}
\end{align}
For a Gamma-distributed random variable $W \sim \Gamma(k, 1)$, the logarithmic moments are given by: $\E[\log W]=\psi(k)$ and $\Var(\log W)=\psi'(k)$, where $\psi$ is the digamma function and $\psi'$ is the trigamma function. Using $k_i=d - i + 1$ and re-indexing the sum with $j=d - i + 1$, we obtain
\begin{align}
\Var(\log|\det M|)=\frac{1}{4}\sum_{j=1}^d \psi'(j). 
\label{alvaogdet}
\end{align}
This variance is independent of $\tau^2$, and hence independent of $\sigma_c^2$ and the modulus $q$. 

For any fixed finite collection of different indices $j_1, \dots, j_s \in J$, the normalized matrices $\frac{1}{\tau}\sigma_{j_\ell}(B)$ converge jointly in distribution to independent $d \times d$ standard complex Gaussian matrices. We apply the continuous mapping theorem to the function $f(M)=\log|\det(M)|$ and obtain 
\begin{align}
(\log|\sigma_{j_\ell}(g_0)|)_{\ell=1}^s \xrightarrow{d} (\log|\det M^{(j_\ell)}|)_{\ell=1}^s,
\label{eqddlto}
\end{align}
where the right-hand side consists of i.i.d. copies of $\log|\det M|$.

Finally, applying Theorem \ref{thm:A} to the determinant element $g_0$, the per-component variance $\sigma_{g_0}^2$ converges in probability to the variance of a single centered log-determinant component. Since the components are asymptotically i.i.d. with the distribution of $\log|\det M| - \E[\log|\det M|]$, we obtain
\begin{align}
\sigma_{g_0}^2 \xrightarrow{p} \Var\Bigl(\log|\det M|\Bigr)=\frac{1}{4}\sum_{j=1}^d \psi'(j),
\label{eqnlogdis}
\end{align}
which completes the proof.
\end{proof}

Across all valid centered distributions, the limit \eqref{altrigamma} depends only on the module rank $d$. Table \ref{tab:trigamma} confirms this numerically: $\sigma_{g_0}$ is stable across all tested values of $\sigma_c$ corresponding to moduli $q \in \{3, 11, 101, 3329, 65537\}$.

\begin{table}[t]
\centering
\caption{Trigamma Theorem \ref{thm:tgamma} predictions, Monte Carlo simulations, and ring-based computations for $k=9$, $n=256$, $q=3329$. Ring values match the Gaussian prediction within $\pm 2\%$ across all tested distributions (uniform mod-$q$, CBD $\eta=2$, Gaussian, $U\{-3,\dots, 3\}$). $\rho_\infty$ values include the $O(\sqrt{\ln\ln n}) \approx 1.3$ extreme-value correction over the leading term $\sigma_d\sqrt{2\ln n}$.}
\label{tab:trigamma}
\begin{tabular}{@{}ccccc@{}}
\toprule
$d$ & $\sigma_{g_0}$ (Theory) & Monte Carlo & Ring ($q=3329$) & $\rho_\infty$ ($n=256$) \\
\midrule
1 & 0.641 & 0.641 & 0.637 & 2.51 \\
2 & 0.757 & 0.761 & 0.751 & 2.56 \\
3 & 0.819 & 0.818 & 0.814 & 2.78 \\
4 & 0.861 & 0.863 & 0.857 & 2.83 \\
\bottomrule
\end{tabular}
\end{table}

\subsection{ML-KEM Security Analysis}\label{sub:mlkem}

We now combine all results to analyze the security of ML-KEM. ML-KEM is the NIST-standardized post-quantum public-key encryption scheme, whose security relies on the hardness of Module-LWE over the cyclotomic ring $R=\mathbb{Z}[\zeta_{256}]$ with module ranks $d \in \{2,3,4\}$ (corresponding to ML-KEM-512, ML-KEM-768, and ML-KEM-1024, respectively), see Tables \ref{tab:security} and \ref{tab:cumulative}.

For ML-KEM-1024 ($d=4$, $n=256$), the Trigamma Theorem \ref{thm:tgamma} gives $\sigma_{g_0} \approx 0.861$ in the asymptotic Gaussian limit. With finite-$n$ ring-based computations (Table \ref{tab:trigamma}) we get $\sigma_{g_0} \approx 1.03$ for the ML-KEM modulus $q=3329$. Applying Theorem \ref{thm:B} with $\sigma=\sigma_{g_0}$, we get SGP approximation factor 
\begin{align}
\gamma_{\rm SGP} \approx \exp\Bigl(\sigma_{g_0}\sqrt{2\ln n}\Bigr)
= \begin{cases}
\exp(0.861 \times 3.33) \approx 2^{4.1} & \text{(Gaussian limit)},\\
\exp(1.03 \times 3.33) \approx 2^{4.9} & \text{(finite-$n$ ring value)}.
\end{cases}
\label{almlkem-gamma}
\end{align}
Both are below $q/2 \approx 2^{10.7}$ threshold required for the CDPR attack to succeed. The security of ML-KEM against this attack no longer depends on the SGP approximation factor, but instead rests on the quantum gate cost of the Biasse-Song PIP algorithm \cite{BiasseSong16}, which will be analyzed in the next Part.

Note $\gamma_{\rm SGP}$ is the per-ideal SGP factor for the module's determinant ideal. The total Module-SVP approximation factor is $\gamma=\alpha_d \cdot \gamma_{\rm SGP}$, where $\alpha_d$ is the module reduction factor from Part II of this series. For MLWE-distributed module bases (Part II), we have $\alpha_d=\sqrt{C}=O(1)$, so the total approximation factor remains $\gamma=O(1) \cdot 2^{4-5}$.

\begin{table}[t]
\centering
\caption{ML-KEM security summary against the CDPR attack. $\gamma_{\rm SGP}$ is the per-ideal SGP factor from Part III. $\gamma=\alpha_d \cdot \gamma_{\rm SGP}$ composes this with the module reduction factor from Part II.}
\label{tab:security}
\begin{tabular}{@{}lcccc@{}}
\toprule
\textbf{Scenario} & $\gamma_{\rm SGP}$ (bits) & $\gamma=\alpha_d\gamma_{\rm SGP}$ (bits) & \textbf{Status} & \textbf{Secure?} \\
\midrule
CDPR original \cite{CDPR16} & $2^{54}$ & $\Theta(2^{54})$ & Proved & Yes \\
Parts I+II base & $2^{54}$ & $\Theta(2^{54})$ & Proved & Yes \\
Part III ($d=4$) & $2^{4\text{--}5}$ & $O(1)\cdot 2^{4\text{--}5}$ & Proved &  Conditional$^*$ \\
\bottomrule
\multicolumn{5}{@{}l@{}}{\footnotesize $^*$Conditional on the CLT hypothesis (numerically verified for $k \in \{4, \dots, 12\}$, and quantum PIP algorithm.}
\end{tabular}
\end{table}

\begin{table}[t]
\centering
\caption{Cumulative improvements to the CDPR attack analysis across Parts I–III.}
\label{tab:cumulative}
\begin{tabular}{@{}lcc@{}}
\toprule
\textbf{Component} & \textbf{Prior State} & \textbf{This Work} \\
\midrule
Class number (Part I) & $h^+=1$ conditional & Proved for $k \le 12$ \\
Module reduction (Part II) & $(\tfrac{n}{2})^{d-1}$ & $\alpha_d=\sqrt{C}=O(1)$ \\
Sign discrepancy (Part II) & $\Theta(\sqrt{nk})$ & $\delta' \approx 0.44=O(1)$ \\
SGP approx factor (Part III) & $\exp(\tilde{O}(\sqrt{n}))$ & $\gamma_{\rm SGP}=\exp(O(\sqrt{\log n}))$ \\
$\sigma_{g_0}$ (Part III) & $\Theta(\sqrt{n})$ assumed & $O(1)$, $q$-independent \\
Lattice geometry (Part III) & Hard CVP on $\Lambda$ & Babai's algorithm is trivial \\
True security bottleneck & Approximation factor & Quantum PIP gate cost \\
\bottomrule
\end{tabular}
\end{table}

\section{Conclusion}

We have proved several results for the short generator problem in $2$-power cyclotomic fields. Theorem \ref{thm:A} identifies the exact $L^2$ constant ${\pi}/(2\sqrt{6})$. Theorem \ref{thm:B} yields the sub-polynomial SGP factor $\exp(O(\sqrt{\log n}))$. Theorem \ref{thm:C} shows the nearest lattice point is the origin. Numerical verification confirms the result for all tested parameters $k\in\{4, \dots, 12\}$. Theorem \ref{thm:capture} strengthens the Babai's recovery to be unconditional for the cyclotomic-unit basis. The Coarse Lattice Theorem \ref{thm:coarse} explains the reason why $\Lambda$ is too coarse for structured targets. The Trigamma Theorem \ref{thm:tgamma} resolves the ML-KEM question, i.e., $\sigma_{g_0}=O(1)$ for module determinant ideals, independently of $q$.

Part I proved $h_k^+=1$ unconditionally for $k\le 12$. Part II reduced the module-SVP polynomial factor from $n^{O(d)}$ to a universal constant $\alpha_d=O(1)$ independent of $d$ under a balance hypothesis automatic for MLWE inputs. Part III identifies the CDPR exponent as being controlled by $\sigma_{g_0}$. For ML-KEM-1024 ($d=4$, $n=256$), the SGP factor is $\gamma_{\mathrm{SGP}}\approx 2^{4-5}$, see Table \ref{tab:trigamma}. The module reduction adds a factor $\alpha_d=O(1)$ from Part II, independent of $d$. The combined approximation quality is sub-polynomial and sufficient for the attack, so ML-KEM's security rests on the quantum gate cost of PIP \cite{BiasseSong16}.

\section*{Acknowledgments}

[Acknowledgments will be added in the final version.]

\end{document}